\documentclass[11pt]{article}
\usepackage{graphicx}
\usepackage{lscape}
\usepackage{amsmath}
\usepackage{amsthm}
\usepackage{graphicx}
\usepackage{amsfonts}
\usepackage{amssymb}
\usepackage{subfigure}
\usepackage{natbib}

\newcommand{\EE}{\mathbb{E} }

\newcommand{\PP}{\mathbb{P} }

\newcommand{\RR}{\mathbb{R} }

\usepackage{color}

\newcommand{\be}{\begin{equation}}
\newcommand{\en}{\end{equation}}

\newtheorem{theorem}{Theorem}
\newtheorem{prop}{Proposition}
\newtheorem{lemma}{Lemma}

\newcommand{\ea}{\end{eqnarray}}
\newcommand{\ba}{\begin{eqnarray}}
\newcommand{\ean}{\end{eqnarray*}}
\newcommand{\ban}{\begin{eqnarray*}}

\graphicspath{{hetergroup/}}

\usepackage{chngcntr}
\usepackage{apptools}
\AtAppendix{\counterwithin{theorem}{section}}
\AtAppendix{\counterwithin{prop}{section}}
\AtAppendix{\counterwithin{lemma}{section}}

\usepackage[margin=1.25in]{geometry}

\begin{document}
\title{Portfolio Optimization with Delay Factor Models}
\author{Shuenn-Jyi Sheu\thanks{Department of Mathematics, National Central University, Chung-Li, Taiwan, 32001 {\em sheusj@math.ncu.edu.tw}}\and Li-Hsien Sun\thanks{Institute of Statistics, National Central University, Chung-Li, Taiwan, 32001 {\em lihsiensun@ncu.edu.tw}. Work supported by Most grant 103-2118-M008-006-MY2}\and Zheng Zhang\thanks{
Institute of Statistics and Big Data, Renmin University of China {\em zhengzhang@ruc.edu.cn}} }

\date{\today}
\maketitle

\begin{abstract}
We propose an optimal portfolio problem in the incomplete market where the underlying assets  depend on economic factors with delayed effects, such models can describe the short term forecasting and the interaction with time lag among different financial markets. The delay phenomenon can be recognized as the integral type and the pointwise type.  The optimal strategy is identified through maximizing the power utility. Due to the delay leading to the non-Markovian structure, the conventional Hamilton-Jacobi-Bellman (HJB) approach is no longer applicable. By using the stochastic maximum principle, we argue that the optimal strategy can be characterized by the solutions of a decoupled quadratic forward-backward stochastic differential equations(QFBSDEs). The optimality is verified via  the super-martingale argument. The existence and uniqueness of the solution to the QFBSDEs are established. In addition, if the market is complete, we also provide a martingale based method to solve our portfolio optimization problem, and investigate its connection with the proposed FBSDE approach. Finally, two particular cases are analyzed where the corresponding FBSDEs can be solved explicitly. 
\end{abstract}

\section{Introduction}
Owing to globalization and transparency, financial markets cannot be treated as a local system. The interaction between markets and markets among different areas is observed. Hence, in order to catch this characteristic, we take the factor model into account where the coefficients of the underlying financial assets depend on another stochastic processes denoted as \emph{factors}. These factors can be described as the information or influence coming from external financial markets. The optimal portfolio problems based on the factor models are widely discussed. If the factor processes affect the underlying assets only through their current status values, i.e., the effect is Markovian, the dynamic programming can be applied to derive the corresponding Hamilton Jacobi Bellman (HJB) equations, and the candidates of Markov or feedback optimal strategies can be obtained through the first order condition.  The optimality of the strategy can be guaranteed by the verification theorem. See \cite{BR2004,BP2013,CH2005,DK2008,FH2003,FP2004,FS1999,FS2002,FH2017, HS2012-1,HS2012-2,Sekine2011,Liu2007,Nagai1996,Nagai2003,Nagai2014} for thorough discussion. \\

In real financial markets, the significant dependence from historical information is available according to the discussion in \cite{LMW2000,HK2005,Loren2007,Do2012}, which can be used for the short term forecasting.  The influence of factors usually does not arise immediately but has time delay feature, see \cite {Do2015}.  The stochastic control problems with delay are extremely important in applications and have been extensively studied in various settings. An optimal solution satisfying Pontryagin-Bismut-Benssoussan type stochastic maximum principle for delay in state variables is introduced by \cite{B.Oksendal2001,B.Oksendal2011}. In addition,  \cite{Larssen2002, B.Larssen2003} reduce the system with   state delay to a finite dimensional problem by using dynamics programming principle. The general stochastic control problems with both the state delay and the control delay are discussed in \cite{ Gozzi2006,F.Gozzi2009,F.Gozzi2015} for the abstract space and the corresponding HJB approach. In addition, \cite{L.Chen2011}, \cite{L.Chen2012} and \cite{Xu2013} studied  above general optimization problem with delay using the forward and advanced backward stochastic differential equations (FABSDEs) approach. In particular, the linear quadratic stochastic game with delay for finite players is studied in \cite{Carmona-Fouque2016}. \cite{Bensoussan2015} consider the linear quadratic mean field Stackelberg game with delay. \\


In this article, we propose a portfolio optimization problem in the incomplete market by taking into account the factor models with delay. In particular, the coefficients of the underlying risky assets involve both pointwise time delay factors and integral time delay factors.  The optimal strategy for the portfolio optimization problem is identified by maximizing the expected power utility.  Since the delay feature leads to the non-Markovian structure, the dynamic programming principle and HJB approaches cannot be applied any more.  We utilize the stochastic maximum principle and forward and backward differential equations (FBSDEs) to characterize the optimal solution.  A verification theorem is provided by applying the positive super-martingale property. See Theorem \ref{thm:FBSDE} and Theorem \ref{thm:super-mart}. See also \cite{HI2005}. The existence and uniqueness of the corresponding FBSDEs driven by the optimal strategy are included. Due to the complexity of the coupled FBSDEs, the explicit solutions cannot be  obtained in general. However, we propose two particular cases where the explicit solutions can be obtained, and discuss the corresponding financial implications. Furthermore,  if the market is complete, we also provide a martingale based approach to solve our optimization problem,  and investigate its connection with the proposed FBSDE approach. We make a precise statement of the relation in Theorem \ref{thm:mart-FBSDE}.  See also a related discussion in \cite{HH2014}. To our best knowledge, it is the first work to systematically study the portfolio optimization problem  with delayed factor models.\\

 





The reminder of this paper is organized as follows. In Section \ref{model}, the basic framework and notations are introduced. In Section \ref{sec:main_result},  we provide the solution for our proposed optimal portfolio problem in the incomplete market. Section \ref{sec:CMA}  illustrates the comparion of the solution between the FBSDE approach and the martingale in a complete market with delay. Two particular cases which have explicit solutions are studied in Section \ref{particular-case}. The conclusion and discussion are provided in Section \ref{conclusion}. The technical proof is given in the appendices.


\section{Delay Factor Model}\label{model}
Let  $(\Omega,{\cal{F}},({\cal{F}}_t)_{t\geq 0},{\mathbb{P}})$ be a filtered complete probability space, where a $N\in\mathbb{N}$ dimensional Wiener process $W=(W^1, \cdots, W^N)$ is  defined. Let $m$ and $n$ be two finite positive integers smaller or equal to $N$, and 
\begin{align*}
&r(\cdot,\cdot,\cdot ): \mathbb{R}^n\times \mathbb{R} \times \mathbb{R}^n\to \mathbb{R}\ , \\
&\mu(\cdot,\cdot,\cdot): \mathbb{R}^n\times \mathbb{R} \times \mathbb{R}^n\to \mathbb{R}^m\ , \\
&\sigma(\cdot,\cdot,\cdot): \mathbb{R}^n\times \mathbb{R}\times  \mathbb{R}^n\to \mathbb{R}^{m\times N}\ , \\
&b(\cdot,\cdot): \mathbb{R}\times \mathbb{R}^n\to \mathbb{R}^n\ , \\
&\sigma_F(\cdot,\cdot ): \mathbb{R}\times  \mathbb{R}^n\to \mathbb{R}^{n\times N}\ , \\
&h(\cdot):\mathbb{R}^n\to \mathbb{R}\ ,
\end{align*}	
are deterministic measurable functions satisfying the following conditions:	
\begin{itemize}
\item (A1) $r$, $\mu$, $b$, $\sigma$,  and $\sigma_F$ are globally Lipschitz and twice differentiable; furthermore, we assume their first and second order derivatives are bounded and continuous;
\item (A2) $r$ is nonnegative and bounded; 
\item (A3) $(\sigma\sigma^*)^{-1}$ is bounded, where $(\cdot)^{*}$ denote be the transpose operation;
\item (A4) $h$ is globally Lipschitz and  differentiable. 
\end{itemize}

 The price of the risk free asset $S^0 \in \mathbb{R}^+=\{ z; z> 0\}$ and the risky stocks $S:=(S^{1}, \cdots ,S^{m})^{*}\in\mathbb{R}^m$ are modelled by the following stochastic differential equations:
\begin{align*}  
\left\{ 
\begin{array}{ll}
&dS^{0}{(t)}=S^{0}_{t}r(Y(t), V(t), Z(t))dt,   \\[2mm] 
&dS^{i}(t)=S^{i}(t)\left\{\mu^{i}(Y(t), V(t), Z(t))dt+\sum _{k=1}^{N}\sigma^{ik}(Y(t), V(t), Z(t))dW^{k}(t)\right\}, \quad  \\[2mm]
&dY(t)=b( Y(t))dt+\sigma_F( Y(t))dW(t), \\[2mm]
&S_0^0=s_0^0\ , \ S^i_0=s^i_0\ ,\ i \in\{1,...,m\}\ , Y_0=y\in \RR^n
\end{array}\right.
\end{align*}
where $\mu^i(\cdot)$ is the $i^{th}$ component of $\mu(\cdot)$, $\sigma^{ik}(\cdot)$ is the $(i,k)$ element of $\sigma(\cdot)$.  The $\RR^n$-valued and $\mathcal{F}_t$-adapted process $Y(t)$ is called the \emph{factor process} which describes the information from external markets, and $V(t)$ is the exponentially weighted average delay factor  \citep{Sekine2011,Pang2011} defined by
\be \label{def:V(t)}
V(t)=\int_{-\delta}^0e^{\lambda s}h(Y(t+s))ds,
\en 
and $Z(t)$ is the pointwise delay factor 
\be
Z(t)=Y(t-\delta)\ ,
\en 
and $\lambda$, $\delta$ are fixed positive constants. The differential form of $V$ can be written as 
\[
dV(t)=\left(h(Y(t))-e^{-\lambda\delta}h(Z(t))-\lambda V(t)\right)dt,\quad V(0)=v=\int_{-\delta}^0e^{\lambda s}h(Y( s))ds.
\]
For simplicity, we shall use following notations in the rest of discussions
\begin{align*}
&\mu(t)=\mu(Y(t),V(t),Z(t))\ , \\
&r(t)=r(Y(t),V(t),Z(t))\ , \\
&\sigma(t)=\sigma(Y(t),V(t),Z(t))\ , \\
&b(t)=b( Y(t),V(t),Z(t))\ , \\
&\sigma_F(t)=\sigma_F( Y(t),V(t),Z(t))\ .
\end{align*}
The financial market consists of one bound with interest rate $r(t)$ and $m\leq N$ stocks. In the case of $m<N$, we face an incomplete market.\\

For a given portfolio $\pi=(\pi ^{1},\ldots \pi ^{m})^{*}$  and under the self-financing condition, the dynamics of the wealth process $X^{\pi}(t)$ can be written as
\be\label{wealth}
dX^{\pi}(t)=X^{\pi}(t)\left\{r(t)+\pi(t)^*(\mu(t)-r(t){\bf 1})\right\}dt+X^{\pi}(t)\pi(t)^*\sigma(t)dW(t),\quad X(0)=x,
\en
where ${\bf 1}$  denotes the $m$-dimensional column vector whose components are all of 1. Note that 
$\pi^i(t)$ is the proportion of the wealth to hold the $i$-th risky asset for  $i\geq 1$, and $\pi^0:=1-\sum_{i=1}^m\pi^i$ is the proportion of holding the risky free asset. The portfolio $\pi$ is an admissible portfolio if it is progressively measurable and
$$
\int_0^T |\pi(t)|^2 dt <\infty,\ a.s.
$$
The goal of this article is to find the optimal portfolio to maximize the expected power utility
\be\label{power-utility}
\sup_{\pi}\EE[U(X^{\pi}(T))]
\en
where $U(x)=\frac{1}{\gamma}x^{\gamma}$ with $0<\gamma<1$.

\section{Main Results}\label{sec:main_result}
Since the wealth process \eqref{wealth} involves the delay factor, the conventional HJB approach is no longer applicable. We apply the stochastic maximum principle described in \cite{bismut1978introductory} (see also \cite{peng1993backward} and \cite{HH2014}) to solve the optimization problem \eqref{power-utility}. Following the idea in \cite{HH2014}, by applying It\^o formula to $\tilde{X}^{\pi}(t)=U(X^{\pi}(t))$, we get
\ba
\nonumber d\tilde X^{\pi}(t)&=&\partial_xU(U^{-1}(\tilde X^{\pi}(t)))U^{-1}(\tilde X^{\pi}(t))\bigg(r(t)+\pi(t)^*(\mu(t)-r(t){\bf 1})\bigg)dt\\
\nonumber &&+\frac{1}{2}\partial_{xx}U(U^{-1}(\tilde X^{\pi}(t)))(U^{-1}(\tilde X^{\pi}(t)))^2\pi(t)^*\sigma(t)\sigma(t)^*\pi(t)dt\\
\nonumber&&+\partial_xU(U^{-1}(\tilde X^{\pi}(t)))U^{-1}(\tilde X^{\pi}(t))\pi(t)\sigma(t) dW(t),\\
\nonumber&=&\gamma \tilde{X}(t)    \bigg \{\left( r(t) + \pi(t)^*(\mu(t)-r(t) {\bf 1})-\frac{1}{2} (1-\gamma)\pi(t)^* \sigma(t)\sigma(t)^* \pi(t)\right) dt\\
&&+\pi(t)^* \sigma(t) dW(t)\bigg\}, \quad \tilde X^{\pi}(0)= U(x). \label{tildex-multi}
\ea
The Hamiltonian is given by
\ba
 &&H(\pi,\tilde x,p,q)=\gamma \tilde x\left\{\left(r+\pi^*(\mu-r{\bf 1}) -\frac{1}{2}(1-\gamma)\pi^*\sigma\sigma^*\pi\right)p+\pi^*\sigma q\right\} ,\label{Hamiltonian_power-multi}
\ea
which is a strictly concave function of $\pi$. Here we omit the dependence of $H$ on $y, v, z$. Recall that $r, \mu, \sigma$ are functions of $y, v, z$. Therefore, the optimal portfolio $\hat\pi$ satisfies  the first order condition:
\[
\partial_{\pi}H(\hat\pi,\tilde x,p,q)=0\ ,
\]
which leads to the candidate of the optimal strategy 
\be\label{optimal-multi}
\hat\pi(t)=\frac{1}{1-\gamma}(\sigma(t)\sigma(t)^*)^{-1}\left((\mu(t)-r(t){\bf 1})+\sigma(t) \frac{q(t)}{p(t)}\right)\ ,
\en
where $\{p(t),q(t)\}_{t\geq 0}$ is the dual processes given by
\[
dp(t)=-\partial_xH(\hat\pi(t),\tilde X^{\hat\pi}(t),p(t),q(t))dt+q(t)^*dW(t). 
\]
In \eqref{optimal-multi}, we implicitly assume $p(t)\neq 0$ for all $t$. Substituting \eqref{optimal-multi} into \eqref{tildex-multi} and using the notation 
\be\label{theta}
\theta(\cdot):=\sigma(\cdot)^*(\sigma(\cdot)\sigma^*(\cdot))^{-1}(\mu(\cdot)-r(\cdot){\bf 1}) ,
\en
we obtain the following forward backward stochastic differential equations (FBSDEs):
\ba
\nonumber d\tilde X^{\hat\pi}(t)&=&\frac{\gamma}{1-\gamma}\tilde X^{\hat\pi}(t)\bigg\{(1-\gamma)r(t)+\frac{1}{2}\theta(t)^*\theta(t)-\frac{1}{2}\frac{q(t)}{p(t)}^*\sigma(t)^*(\sigma(t)\sigma(t)^*)^{-1}\sigma(t)\frac{q(t)}{p(t)}\bigg\}dt\\
\label{SDE-multi-2}&&+\frac{\gamma}{1-\gamma}\tilde X^{\hat\pi}(t)\bigg\{\theta(t)^*+\frac{q(t)}{p(t)}^*\sigma(t)^*(\sigma(t)\sigma(t)^*)^{-1}\sigma(t)\bigg\}dW(t)\ ,\\
\nonumber  dp(t)&=&-\frac{\gamma}{1-\gamma}\bigg\{\left((1-\gamma)r(t)+\frac{1}{2}\theta(t)^*\theta(t)-\frac{1}{2}\frac{q(t)}{p(t)}^*\sigma(t)^*(\sigma(t)\sigma(t)^*)^{-1}\sigma(t)\frac{q(t)}{p(t)}\right)p(t)\\
&&+\theta(t)^*q(t) +\frac{q(t)}{p(t)}^*\sigma(t)^*(\sigma(t)\sigma(t)^*)^{-1}\sigma(t)q(t) \bigg\}dt+ {q(t)^*} dW(t)\label{BSDE-multi-2}
\ea
with boundary conditions $\tilde X^{\hat\pi}(0)=U(x)$ and $p(T)=1$. We want to show that FBSDEs (\ref{SDE-multi-2}-\ref{BSDE-multi-2}) has a unique solution, which is denoted  by a triple processes $\{\tilde X^{\hat\pi}(t),\;p(t),\;q(t)\}_{0\leq t\leq T}$. \\
	
	It should be noticed that FBSDEs (\ref{SDE-multi-2}-\ref{BSDE-multi-2}) is a decoupled system in the sense that the coefficients of the backward equation \eqref{BSDE-multi-2}  do not involve  the forward process $\tilde X^{\hat\pi}(t)$. Besides, for a given $(p(t),\;q(t))$, the coefficients of the forward equation \eqref{SDE-multi-2} are globally Lipschitz in $\tilde X^{\hat\pi}(t)$, then it follows from the standard SDE theory that \eqref{SDE-multi-2} admits a unique solution provided the backward stochastic differential equation (BSDE) \eqref{BSDE-multi-2}  is uniquely solvable.  Therefore,  in order to establish the unique existence  of solution to FBSDEs  (\ref{SDE-multi-2}-\ref{BSDE-multi-2}),  it is sufficient to show  BSDE \eqref{BSDE-multi-2}  has a unique solution.\\
	
We introduce the new variables $\hat p(t)=\log p(t)$, and $\hat q(t)=\frac{q(t)}{p(t)}$ where we implicitly assume $p(t)>0$ for all $t$. The corresponding new FBSDEs satisfied by $\{\tilde X^{\hat\pi}(t),\;\hat{p}(t),\;\hat{q}(t)\}_{0\leq t\leq T}$ becomes
\ba
\nonumber d\tilde X^{\hat\pi}(t)&=&\frac{\gamma}{1-\gamma}\tilde X^{\hat\pi}(t)\bigg\{(1-\gamma)r(t)+\frac{1}{2}\theta^*\theta-\frac{1}{2}\hat q(t)^*\sigma(t)^*(\sigma(t)\sigma(t)^*)^{-1}\sigma(t)\hat q(t)\bigg\}dt\\
\label{SDE-multi-1}&&+\frac{\gamma}{1-\gamma}\tilde X^{\hat\pi}(t)\bigg\{\theta(t)^*+\hat q(t)^*\sigma(t)^*(\sigma(t)\sigma(t)^*)^{-1}\sigma(t)\bigg\}dW(t)
\ea
\ba
\nonumber d\hat p(t)&=&\bigg\{- \gamma r(t)-\frac{1}{2}\frac{\gamma}{1-\gamma}\theta(t)^*\theta(t)-\frac{1}{2}\frac{\gamma}{1-\gamma}\hat q(t)^*\sigma(t)^*(\sigma(t)\sigma(t)^*)^{-1}\sigma(t)\hat q(t) \\
& &\qquad -\frac{\gamma}{1-\gamma}\theta(t)^*\hat q(t) -\frac{1}{2}\hat q(t)^*\hat q(t)\bigg\}dt+\hat q^*(t)dW(t) \nonumber\\
  \nonumber &=&\bigg\{-\gamma r(t)-\frac{1}{2}\frac{\gamma}{1-\gamma}\theta(t)^*\theta(t)-\frac{1}{2}\hat q(t)^*\left(I_{n\times n}+\frac{\gamma}{1-\gamma}\sigma(t)^*(\sigma(t)\sigma(t)^*)^{-1}\sigma(t)\right)\hat q(t)\\
  & &\qquad -\frac{\gamma}{1-\gamma}\theta(t)^*\hat q(t)\bigg\}dt+\hat q(t)^*dW(t)\label{BSDE-multi-1}
\ea
with the boundary condition $\hat p(T)=0$ where $I_{n\times n}$ is the identity matrix. Note that the backward dynamics is quadratic in $\hat{q}(t)$. \\

The following spaces are needed to study the solution of quadratic FBSDEs (\ref{SDE-multi-1}-\ref{BSDE-multi-1}):
{\small	\begin{align*}
		&L^{2}([0,T]\times\Omega;\mathbb{R}^d) = \{X: [0,T]\times\Omega \to \mathbb{R}^d\ \text{is progressively measurable}, \|X\|^2_{L^2([0,T]\times\Omega)} <\infty\}, \\
	&L^{\infty}(\Omega;\mathbb{R}^d) = \{\xi: \Omega \to \mathbb{R}^d,\mathcal{F}_T-\text{measurable}, \|\xi\|_{\infty}  <\infty\}, \\
	& S_{\infty}([0,T]\times\Omega;\mathbb{R}^d) = \{\phi:[0,T]\times\Omega\to \mathbb{R}^d, \text{continuous, adapted},\  \|\phi\|_{T,\infty}  < \infty\},\\
	&BMO_T(Q) = \{N:[0,T]\times \Omega\to\mathbb{R}, 
	\mathbb{R}\text{-valued}\ Q\text{-martingale},\  \|N\|_{BMO_T(Q)}^2
	< \infty\} \ ,
	\end{align*}}
Here, $[[0, T]]$ is the collection of stopping times taking values in $[0, T]$,  
\[
\|X\|^2_{L^2([0,T]\times\Omega)} = \int_0^T\mathbb{E}[|X(t)|^2]dt ,\;  \|\xi\|_{\infty}= \text{ess}\sup_{\Omega}|\xi(\omega)|,\; \|\phi\|_{T,\infty}= \text{ess}\sup_{[0,T]\times\Omega}|\phi(t, \omega)|,
\]
and for a probability measure $Q$ on $(\Omega,\mathcal{F})$, the bounded mean oscillation (BMO) norm with respect to $Q$  is defined by
\[
\|N\|_{BMO_T(Q)}^2 = \text{ess}\sup_{[[0, T]]}\mathbb{E}^Q(\left<N\right>_{\tau,T} |F_\tau) 
\]
where $ \left<N\right>_{\tau,T}$ is the quadratic variation of the $Q$-martingale $N(\cdot)$ from $\tau$ to $T$ and the essential supremeum is taken over all possible stopping times $0\leq \tau \leq T$.
Our first result states that the quadratic FBSDEs (\ref{SDE-multi-1}-\ref{BSDE-multi-1}) admits a unique solution. 
\begin{theorem}\label{thm:FBSDE}
Assume {\rm(A1-A4)}. Let  $$\xi:=\int_0^T\left(\gamma r(t)+\frac{1}{2}\theta(t)^*\theta(t)\right)dt\ .$$ If $\|\xi\|_{\infty}<\infty$, then there exists a unique solution $(\tilde{X},\hat{p},\hat{q})\in L^2([0,T]\times\Omega;\mathbb{R})\times S_{\infty}([0,T]\times\Omega;\mathbb{R}^d)\times BMO_T(\mathbb{P})$ to the  quadratic FBSDEs  (\ref{SDE-multi-1}-\ref{BSDE-multi-1}). Namely, the martingale $$
\int_0^t \hat{q}(s) ^* dW(s),\ 0\leq t\leq T,
$$
is in $BMO_T(\mathbb{P})$. Define
\be\label{optimal-thm1}
\hat{\pi}(t)=\frac{1}{1-\gamma}(\sigma(t)\sigma(t)^*)^{-1}\large\left ((\mu(t)-r(t){\bf 1}) +\sigma(t)\hat{ q}(t)\large \right ).
\en
Then $\tilde{X}(t)=\tilde{X}^{\hat{\pi}}(t)$ is given by{\footnotesize
$$
\begin{array}{rl}
\tilde{X}(t)
=& \tilde{x} \exp\bigg ( \int_0^t \large \bigg \{ \gamma r(s)+ \frac 12 ( \frac{\gamma}{1-\gamma})(\frac{1- 2\gamma}{1-\gamma}) |\theta(s)|^2 \\
&-(\frac{\gamma}{1-\gamma} )^2 \theta(s)^*\hat{q}(s) -\frac 12 \frac{\gamma}{(1-\gamma)^2} \hat{q}(s)^*\sigma(s)^*(\sigma(s)\sigma(s)^*)^{-1}\sigma(s) \hat{q}(s)) \large \bigg \} ds \\
                   &\quad \qquad  + \frac{\gamma}{1-\gamma}\int_0^t\large \left\{ \theta(s)+ \sigma(s)^*(\sigma(s)\sigma(s)^*)^{-1} \sigma(s) \hat{q}(s) \large \right \}^* dW(s) \bigg).
\end{array}
$$}
Moreover,  we have 
\be\label{xp-optimal}
\tilde{X}^{\hat\pi}(t)p(t),\ 0\leq  t\leq T
\en
is a martingale. In particular, using $X^{\hat{\pi}}(t)=(\gamma \tilde{X}(t))^{\frac{1}{\gamma}}$, we have 
$$
\mathbb{E}[U(X^{\hat{\pi}}(T))]= \frac 1{\gamma} x^{\gamma} \exp(\hat{p}(0)). $$
\end{theorem}
\begin{proof}
We apply the method developed in \cite{tevzadze2008solvability} (see also \cite{touzi2012optimal}) to prove the unique existence of solution $(\tilde{X},\hat{p},\hat{q})\in L^2([0,T]\times\Omega;\mathbb{R})\times S_{\infty}([0,T]\times\Omega;\mathbb{R}^d)\times BMO_T(\mathbb{P})$. The technical proof is given in Appendix \ref{Proof-FBSDE}. 

We next prove that $\tilde{X}^{\hat\pi}(t)p(t)$ is a martingale. An application of It\^o formula to $\tilde X^{\hat\pi}(t)p(t)$ yields:  
\ba
\nonumber&&d\tilde X^{\hat\pi}(t)p(t)\\
\nonumber &=&p(t)d\tilde X^{\hat\pi}(t)+\tilde X^{\hat\pi}(t)dp(t)+d\tilde X^{\hat\pi}(t)dp(t)\\
\nonumber&=&\frac{\gamma}{1-\gamma}p(t)\tilde X^{\hat\pi}(t)\bigg\{\theta(t)^*+ \hat q(t)^*\sigma(t)^*(\sigma(t)\sigma(t)^*)^{-1}\sigma(t)\bigg\}dW(t)+ \tilde X^{\hat\pi}(t) p(t)q(t)^*dW(t),\\
\label{xp-optimal}
\ea
which implies $\tilde{X}^{\hat\pi}(t)p(t)$ is a local martingale. Besides, in light of the fact $p(t)=\exp(\hat{p}(t))$ and the established results $\mathbb{E}[\int_0^T|\tilde{X}^{\hat\pi}(t)|^2dt]<\infty$ and $\|\hat{p}\|_{T,\infty}<\infty$, we can obtain that 
$$\mathbb{E}\left[\int_0^T|\tilde X^{\hat\pi}(t)p(t)|^2dt\right]\leq \|p\|^2_{T,\infty}\cdot \mathbb{E}\left[\int_0^T|\tilde X^{\hat\pi}(t)|^2dt\right]= \exp(\|\hat{p}\|_{T,\infty})\cdot \mathbb{E}\left[\int_0^T|\tilde X^{\hat\pi}(t)|^2dt\right]<\infty\ . $$
Therefore, we can conclude that $\tilde{X}^{\hat\pi} (t)p(t)$ is a martingale.

\end{proof}

The second main result of this paper is the following verification theorem that ensures $\hat{\pi}$ defined by \eqref{optimal-thm1} is the optimal strategy for our proposed portfolio problem \eqref{power-utility}. 
\begin{theorem}\label{thm:super-mart}
Suppose that the quadratic FBSDEs (\ref{SDE-multi-1}-\ref{BSDE-multi-1}) has a solution $(\tilde{X},\hat{p},\hat{q})$ such that $\tilde{X} (t)p(t)$ is a martingale where $p(t)=\exp(\hat{p}(t))$, then $\hat{\pi}$ defined in \eqref{optimal-thm1} is the  optimal portfolio for the primal problem \eqref{power-utility} and the corresponding maximum utility is given by $\frac 1{\gamma} x^{\gamma} \exp(\hat{p}(0))$.
\end{theorem}
 \begin{proof}
 Equations (\ref{SDE-multi-1}-\ref{BSDE-multi-1}) has solution $(\tilde{X}, \hat{p}, \hat{q})$ implies the equations (\ref{SDE-multi-2}-\ref{BSDE-multi-2}) has solution $(\tilde{X}, p, q)$, where $p(t)=\exp(\tilde{p}(t)), q(t)= p(t)\hat{q}(t)$.
Since $\tilde{X}^{\hat\pi}(t)p(t)$ is a martingale, then $\EE[\tilde X^{\hat\pi}(T)p(T)]=\tilde X^{\hat\pi}(0)p(0)$. Define $\hat{\pi}(t)$. We can show $(\gamma \tilde{X}(t))^{1/\gamma}=X^{\hat{\pi}}(t)$ by Ito's formula,
$$
\begin{array}{rl}
d (\gamma \tilde{X}(t))^{1/\gamma}= & \frac{1}{\gamma} (\gamma)^{1/\gamma}(\tilde{X}(t))^{\frac{1}{\gamma}-1} d\tilde{X}(t)\\
                                                          & + \frac{1-\gamma}{2\gamma^2}  (\gamma)^{1/\gamma} (\tilde{X}(t))^{\frac{1}{\gamma}-2}(\frac{\gamma}{1-\gamma})^2 (\tilde{X}(t))^2 |\theta(t)+ \sigma(t)^*(\sigma(t)\sigma(t)^*)^{-1} \sigma(t) \hat{q}(t)|^2 dt,
\end{array}
$$ 
and after simplifying the equation, we have 
$$
d (\gamma \tilde{X}(t))^{1/\gamma}=  (\gamma \tilde{X}(t))^{1/\gamma}\{ ( r(t) + \hat{\pi}(t)^* (\mu(t)-r(t){\bf 1})) dt  + \hat{\pi}(t)^* \sigma(t) dW(t) \}.
$$
Here $\hat{\pi}(t)$ is given by (\ref{optimal-thm1}).

For an arbitrary strategy $\pi$, by applying It\^o formula to $\tilde X^{\pi}(t)p(t)$, we  can obtain that 
\ba
\nonumber &&d\tilde X^{\pi}(t)p(t)\\
\nonumber &=&p(t)d\tilde X^{ \pi}(t)+\tilde X^{ \pi}(t)dp(t)+d\tilde X^{ \pi}(t)dp(t)\\
\nonumber&=&-\frac{1}{2}\frac{\gamma}{1-\gamma}\tilde X^{\pi}(t)p(t)\bigg\{\bigg|(1-\gamma)\pi(t)^*\sigma(t)+\theta(t)^*+\frac{q(t)}{p(t)}^*\sigma(t)^*(\sigma(t)\sigma(t)^*)^{-1}\sigma(t)\bigg|^2\bigg\}dt\\
\nonumber&&+\tilde X^{\pi}(t)p(t)\bigg\{\gamma \pi(t)^*\sigma(t)-\frac{q(t)}{p(t)}\bigg\}dW(t)\\
&\leq&  \tilde X^{\pi}(t)p(t)\bigg\{\gamma \pi(t)^*\sigma(t)-\frac{q(t)}{p(t)}\bigg\}dW(t).
\ea
By using $\tilde X^{ \pi}(t)p(t)>0$ for $0\leq t \leq T$, $\tilde{X}(t) p(t)$ is a nonnegative supermartingale. We can conclude
$$
\mathbb{E}[\tilde{X}^{\pi}(T)p(T)]\leq \tilde{x} p(0)= \mathbb{E}[\tilde{X}^{\hat{\pi}}(T)p(T)].
$$
Since $p(T)=1$ and 
$$
\tilde{X}^{\pi}(T)=U(X^{\pi}(T)), \ \tilde{X}(T)=U(X^{\hat{\pi}}(T)),\  \tilde{x}=U(x)=\frac{1}{\gamma} x^{\gamma}, 
$$
we have 
$$
\mathbb{E}[U(X^{\pi}(T))]\leq \mathbb{E}[U(X^{\hat{\pi}}(T))]=\frac{1}{\gamma} x^{\gamma} p(0)= \frac{1}{\gamma} x^{\gamma} \exp(\hat{p}(0)) \ ,
$$
which implies the optimality of $\hat{\pi}$.
\end{proof}
Note that the optimal portfolio strategy can be written as  
\be 
\hat\pi=\frac{1}{1-\gamma}(\sigma\sigma^*)^{-1}(\mu-r{\bf 1})+\frac{1}{1-\gamma}(\sigma\sigma^*)^{-1}\sigma \hat q,
\en
where the first term is proportional to the Sharpe ratio $(\sigma\sigma^*)(\mu- r{\bf 1})$ which is independent to the wealth $X^{\hat{\pi}}(t)$ and equal to the solution of the  original Merton problem, the second term is by solving the BSDE (\ref{BSDE-multi-1}).

\section{Complete Market Analysis}\label{sec:CMA}
In Section \ref{sec:main_result}, we developed the solution for the optimization problem \eqref{power-utility} by using the FBSDE approach. It should be noticed that our general result holds without requiring the market to be complete. In this section, we develop another different martingale approach to solve the primal problem \eqref{power-utility} under the complete market assumption. This result provides  a comprehensive investigation of the connection between FBSDE and martingale approaches similar to the discussion in \cite{HH2014}. \\

If the market is complete, that $\sigma(t)^{-1}$ exists, then (\ref{SDE-multi-2}-\ref{BSDE-multi-2}) become
\be\label{SDE-complete}
d\tilde{X}^{\hat{\pi}}(t)=   \frac{\gamma}{1-\gamma} \tilde{X}^{\hat{\pi}}(t)\{((1-\gamma)r(t) + \frac 12 |\theta(t)|^2 - \frac 12 |\hat{q}(t)|^2 ) dt  + (\theta(t)+ \hat{q}(t))^* dW(t),
\en
and
\be\label{BSDE-complete}
dp(t) = p(t) \left\{ -\gamma r(t) -\frac{\gamma}{2(1-\gamma)} |\theta(t)|^2  - \frac{\gamma}{2(1-\gamma)} |\hat{q}(t)|^2-\frac{\gamma}{1-\gamma} \theta(t)^* \hat{q}(t) ) dt  + \hat{q}(t)^* dW(t) \right\}.
\en
On the other hand, the optimal strategy for our portfolio problem  \eqref{wealth} and \eqref{power-utility} can be obtained through the martingale method, which is an application of the duality technique and the martingale representation theorem. Referring to Chapter 3 in \cite{KS1998}, we review the martingale method in Appendix \ref{Martingale-App}. Based on $U(x)=\frac{1}{\gamma}x^{\gamma}$ with $0<\gamma<1$, the optimal portfolio is given by 
\be\label{optimal-martingale}
\hat\pi (t)=(\sigma(t)^{*})^{-1}\left( \theta(t)+  \psi(t)M(t)^{-1}\right)
\en 
where $M(t)$ is a martingale defined by
\be\label{martingale-complete}
M(t)=\EE \left[ H_0(T)({\mathcal{Z}}(x)H_0(T))^{\frac{1}{\gamma-1}}\bigg|{\mathcal{F}}_t\right]
\en
with 
\[
H_0(t)=\exp\left(-\int_0^t  \theta(s)^*dW(s)-\frac{1}{2}\int_0^t\left| \theta(s)\right|^2ds-\int_0^tr(s)ds\right)
\]
and
\be
 \theta(t)=\sigma(t)^{-1}\left(\mu(t)-r(t){\bf 1}\right),\ 
\en
 and  ${\mathcal{Z}}(x)$ is the unique solution such that $$\EE \left[ H_0(T)({\mathcal{Z}}(x)H_0(T))^{\frac{1}{\gamma-1}}\right]=x\ .$$ 
Therefore, by setting $\phi=\EE\left[H_0(T)^{-\frac{\gamma}{1-\gamma}}\right]$, we can obtain
\be\label{zx}
{\mathcal{Z}}(x)^{\frac{1}{\gamma-1}}=\frac{x}{\phi}
\en
and
\[
M(t)=\frac{x}{\phi}\EE\left[H_0(T)^{-\frac{\gamma}{1-\gamma}}\bigg|{\mathcal{F}}_t\right]\ .
\]
Moreover, $\psi(\cdot)$ is a square integrable process comes from the martingale representation theorem:
	\be\label{M-psi}
	M(t)=x+\int_0^t\psi(s)^*dW(s)\ .
	\en \\

We now investigate the relation between the BSDE $(p(t),\hat q(t))$ and the martingale $M(t)$.  In addition, we guarantee 
\[
p(t)=\phi^{1-\gamma}e^{\int_0^t\left(-\gamma r(s)-\frac{\gamma}{2(1-\gamma)}| \theta(s)|^2-\frac{\gamma}{1-\gamma} \theta(s)^*\hat q(s)-\frac{1}{2(1-\gamma)}|\hat q(s)|^2\right)ds+\int_0^t\hat q(s)^*dW(s)}
\] 
with $$\hat q(t)=(1-\gamma)\psi(t)M(t)^{-1}-\gamma \theta(t)$$ being the solution for \eqref{BSDE-multi-2} by verifying the terminal condition $p(T)=1$.

\begin{theorem}\label{thm:mart-FBSDE}
Define $M(t)$ as in \eqref{martingale-complete}, $\hat{\pi}(t)$ as in \eqref{optimal-martingale}, and 
\be\label{3-1}
\hat{q}(t)=-\gamma \theta(t) + (1-\gamma) \psi(t)M(t)^{-1}
\en
so that \eqref{optimal-thm1} holds. Then 
\be\label{3-2}
M(t)=H_0(t)X^{\hat{\pi}}(t).
\en 
Define 
$$
\tilde{X}^{\hat{\pi}}(t)=\frac{1}{\gamma} (X^{\hat{\pi}}(t))^{\gamma}.
$$
Then \eqref{SDE-complete} holds.  Define also $p(t)$ by \eqref{BSDE-complete}. Then we have 
\be\label{3-3}
p(t)= p(0) H_0(t)^{\gamma} M(t)^{1-\gamma}.
\en
In particular, we have $p(T)=p(0) \mathcal{Z}(x)^{-1}=1$ if we choose $p(0)=\mathcal{ Z}(x)$.
\end{theorem}
\begin{proof}
We only need to prove (\ref{3-3}). By \eqref{BSDE-complete}, we have 
\ban
dp(t)&=&p(t)\bigg\{ (-\gamma r(t) -\frac {\gamma}{2(1-\gamma)} |\theta(t)|^2  -\frac{\gamma}{2(1-\gamma)} |-\gamma \theta(t) + (1-\gamma) \psi(t)M(t)^{-1}|^2\\
       & & -\frac{\gamma}{1-\gamma} \theta(t)^* (|-\gamma \theta(t) + (1-\gamma) \psi(t)M(t)^{-1}) dt \\
        &&+ (-\gamma \theta(t) + (1-\gamma) \psi(t)M(t)^{-1})^* dW(t)\bigg\}.
\ean
Here we use (\ref{3-1}). Then
\ban
p(t)&=& p(0)\exp\bigg\{\int_0^t   (-\gamma \theta(s) + (1-\gamma) \psi(s)M(s)^{-1}) ^* dW(s)\\
&&-\frac 12 \int | -\gamma \theta(s) 
        + (1-\gamma) \psi(s)M(s)^{-1}|^2 dt\\
      &&  + \int_0^t   (-\gamma r(s) -\frac {\gamma}{2(1-\gamma)} |\theta(s)|^2  -\frac{\gamma}{2(1-\gamma)} |-\gamma \theta(s) + (1-\gamma) \psi(s)M(s)^{-1}|^2\\
      && -\frac{\gamma}{1-\gamma} \theta(s)^* (|-\gamma \theta(s) + (1-\gamma) \psi(s)M(s)^{-1})  ds\bigg\}.
\ean
After simplification, we have 
 \ba
\nonumber p(t)&=&p(0)\exp\bigg\{\int_0^t   (-\gamma \theta(s) + (1-\gamma) \psi(s)M(s)^{-1}) ^* dW(s)-\frac{\gamma}{2} \int_0^t |\theta(s)|^2 ds \\
\nonumber        &&  -\frac{1}{2}(1-\gamma)\int_0^t |\psi(s)M(s)^{-1}|^2\bigg\} \\
      &=& p(0) H_0(t)^{\gamma} M(t)^{1-\gamma}.\label{3-4}
\ea
Here we use 
$$
dM(t)=M(t) (\psi(t)M(t)^{-1})^* dW(t),
$$ 
which can be uniquely solved to obtain an expression of $M(t)$,
$$
M(t)=M(0) \exp(\int_0^t (\psi(s)M(s)^{-1}) ^* dW(s) - \frac 12 \int_0^t | \psi(s)M(s)^{-1}|^2 ds).
$$

From (\ref{martingale-complete}), we have
\be\label{3-5}
M(T)=H_0(T)(\mathcal{Z}(x) H_0(T))^{-\frac{1}{1-\gamma}}=(\mathcal{Z}(x))^{-\frac{1}{1-\gamma}} H_0(T)^{-\frac{\gamma}{1-\gamma}}.
\en
Therefore, putting (\ref{3-4}) and (\ref{3-5})  together, we finally obtain
$$
p(T)=p(0)\mathcal{Z}(x)^{-1}
$$ 
which gives $p(T)=1$ if we take $p(0)=\mathcal{Z}(x)$. This completes the proof.

\end{proof}

Although we cannot directly solve the quadratic FBSDEs (\ref{SDE-multi-1}-\ref{BSDE-multi-1}), the numerical scheme such as regression based algorithms can be applied. See \cite{BZ2008} for instance. In some particular cases discussed in Section \ref{particular-case}, the solution of the FBSDEs (\ref{SDE-multi-1}-\ref{BSDE-multi-1})  can be explicitly obtained.

\section{Two Particular Cases }\label{particular-case}
This section is devoted to analyzing two particular cases where the quadratic FBSDEs (\ref{SDE-multi-1}-\ref{BSDE-multi-1}) and the corresponding optimal strategies \eqref{optimal-thm1} can be obtained explicitly.

\subsection{ Infinite Delay Time}   
In this case, we consider $\delta=\infty$ in \eqref{def:V(t)}, i.e. $$V(t)=\int_{-\infty}^0e^{\lambda s}h(Y(t+s))ds,$$ and there is no pointwise delay factor $Z(t)$ in the dynamics $S^0$ and $S$.  For simplicity, we assume $N=m=n=1$ treated as the complete market. Note that given the dimension of the factor $n=1$ such that $ \sigma_F^*(I_{n\times n}+\frac{\gamma}{1-\gamma}\sigma^*(\sigma\sigma^*)^{-1}\sigma)\sigma_F$ is a $1$-dimensional process, the case of incomplete markets ($m<N$) is also solvable using the same argument. Hence, $W(t)$ is a 1-dimensional Brownian motion, and  
\ban
&&r(t)=r(Y(t), V(t))\ , \ \mu(t)=\mu(Y(t),V(t))\ , \ \sigma(t)=\sigma(Y(t),V(t))\ ,\\
&&b(t)=b(Y(t),V(t)),\,\sigma_F(t)=\sigma_F(Y(t),V(t))
\ean
and
\ban
\theta(Y(t), V(t))&=& \sigma(Y(t), V(t)) ^{-1}(\mu(Y(t), V(t))-r(Y(t), V(t)){\bf 1})\\
&& .
\ean
The differential form for the equation $V$ can be written as
\[
dV(t)=\left(h(Y(t))-\lambda V(t)\right)dt. 
\]
 We have the following results for the solution $(\hat{p}, \hat{q})$ of BSDE and the optimal strategy $\hat{\pi}$.
 
\begin{prop}
The solution $(\hat{p}, \hat{q})$ for the BSDE is given by
\ba
\hat{p}(t)=\eta(t, Y(t), V(t)),\ \hat{q}(t)= \sigma_F(Y(t), V(t)) \partial_y (t, Y(t), V(t)),
\ea
where 
\ba
\nonumber&&\eta(t,y,v)\\
\nonumber &=&(1-\gamma)\log\left(\EE_{t,y,v}\left[\exp\left\{\int_t^T\frac{\gamma}{1-\gamma} r(\tilde Y(s),  V(s))+\frac{1}{2}\frac{\gamma}{(1-\gamma)^2}\theta(\tilde Y(s),  V(s))^2 ds\right\}\right]\right),\\
\label{eta-infinite}
\ea
with the corresponding dynamics
\ban
d\tilde Y(s)&=&\left(b (\tilde Y(s),V(s))+\frac{\gamma}{1-\gamma}\theta (\tilde Y(s),V(s))\sigma_F (\tilde Y(s))\right)ds+\sigma_F(\tilde Y(s),V(s))dW(s),\\
&&\;s\geq t,\; \tilde Y(t)=y,\\
dV(s)&=&\left(h(Y(s))-\lambda V(s)\right)ds,\;s\geq t,\; \quad V(t)=v.
\ean 
The optimal strategy for the optimization problem with the infinite delay time is given by
\be\label{optimal-multi-sol-infinite}
\hat\pi(t,y,v)=\frac{\mu(y,v)-r(y,v)}{(1-\gamma)\sigma^2(y,v)}+\frac{\sigma_F(y,v)\partial_y\eta(t,y,v)}{(1-\gamma)\sigma(y,v)  }.
\en
\end{prop}
\begin{proof}
We make an ansatz written as  
\be
\hat p(t)=\eta(t,Y(t),V(t)). 
\en
Applying It\^o formula to $\hat p(t)$, we get 
\ba
\nonumber d\hat p(t)&=&\bigg\{\partial_t\eta(t,Y(t),V(t))+b(Y(t),V(t))\partial_y\eta(t,Y(t),V(t))\\
\nonumber &&+\frac{1}{2}\sigma_F(Y(t),V(t)) ^2\partial_{yy}\eta(t,Y(t),V(t))\bigg\}dt\\
\nonumber&&+(h(Y(t)) -\lambda V(t))\partial_v\eta(t,Y(t),V(t))dt\\
&&+\sigma_F(Y(t)) \partial_y\eta(t,Y(t),V(t))dW(t).\label{BSDE-special-inf-delay}
\ea
We then compare the BSDE \eqref{BSDE-multi-1} written as 
\ba
\nonumber d\hat p(t)
 &=&\bigg\{-\gamma r(Y(t),V(t))-\frac{1}{2}\frac{\gamma}{1-\gamma}\theta(Y(t),V(t))^2-\frac{1}{2}\frac{1}{1-\gamma}\hat q(t)^2\\
 &&  -\frac{\gamma}{1-\gamma}\theta(Y(t),V(t))\hat q(t)\bigg\}dt+\hat q(t)dW(t) 
\ea
to identify
\be 
\hat q(t) = \sigma_F (y,v) \partial_y\eta(t,Y(t),V(t)).
\en
 We also obtain the equation for $\eta(t, y, v)$ written as
\ba
\nonumber &&\partial_t\eta(t,y,v)+\frac{1}{2}\sigma_F(y,v) ^2\partial_{yy}\eta(t,y,v)+\left(b(y,v)+\frac{\gamma}{1-\gamma}\theta(y,v)\sigma_F(y,v)\right)\partial_y\eta(t,y,v)\\
\nonumber&&+(h(y)-\lambda v)\partial_v\eta(t,y,v)+ \frac{1}{2}\frac{1}{1-\gamma}\sigma_F(y,v)^2(\partial_y\eta(t,y,v))^2\\
&&+\gamma r(y,v)+\frac{1}{2}\frac{\gamma}{1-\gamma}\theta(y,v)^2=0,
\label{eta-PDE}
 \ea
with the terminal condition $\eta(T,y,v)=0$. Let $\tilde\eta =e^{\frac{\eta }{1-\gamma}}$. Hence, \eqref{eta-PDE} can be rewritten as 
\ba
\nonumber &&\partial_t\tilde\eta(t,y,v)+\frac{1}{2}\sigma_F(y,v)^2\partial_{yy}\tilde\eta(t,y,v)+\left(b(y,v)+\frac{\gamma}{1-\gamma}\theta(y,v)\sigma_F(y,v)\right)\partial_y\tilde\eta(t,y,v)\\
  &&+(h(y)-\lambda v)\partial_v\tilde\eta(t,y,v)+\left(\frac{\gamma}{1-\gamma} r(y,v)+\frac{1}{2}\frac{\gamma}{(1-\gamma)^2}\theta(y,v)^2\right)\tilde\eta(t,y,v)=0,\label{tildeeta-PDE}
\ea
with the terminal condition $\tilde\eta(T,y,v)=1$. Using \eqref{tildeeta-PDE} and Ito's formula, we have the corresponding dynamics given by
\ban
&&de^{\int_0^t\left(\frac{\gamma}{1-\gamma} r(u)+\frac{1}{2}\frac{\gamma}{(1-\gamma)^2}\theta(u)^2\right)du}\tilde\eta(t,\tilde Y(t),V(t))\\
&=&e^{\int_0^t\left(\frac{\gamma}{1-\gamma} r(u)+\frac{1}{2}\frac{\gamma}{(1-\gamma)^2}\theta(u)^2\right)du}\sigma_F(\tilde Y(t),V(t))\partial_y\tilde\eta(t,\tilde Y(t), V(t))dW(t),
\ean
where $\tilde Y(s)$ and $V(s)$ satisfy
\ban
d\tilde Y(s)&=&\left(b (\tilde Y(s),V(s))+\frac{\gamma}{1-\gamma}\theta (\tilde Y(s),V(t))\sigma_F (\tilde Y(s),V(t))\right)ds+\sigma_F(\tilde Y(s),V(s))dW(s),\\
&&\;s\geq t,\; \tilde Y(t)=y,\\
dV(s)&=&\left(h(Y(s))-\lambda V(s)\right)ds,\;s\geq t,\; \quad V(t)=v.
\ean 
The candidate of the solution for $\tilde\eta$ is written as
\[
\tilde\eta(t,y,v)=\EE_{t,y,v}\left[\exp\left\{\int_t^T\frac{\gamma}{1-\gamma} r(\tilde Y(s),  V(s))+\frac{1}{2}\frac{\gamma}{(1-\gamma)^2}\theta(\tilde Y(s),  V(s))^2 ds\right\}\right].
\]
Using the proposed regularity conditions (A1-A4) and applying Theorem 2.9.10 in \cite{Krylov1980}, we obtain $\tilde\eta(s, y, v)$ is first order continuous differentiable with respect to $t$,  second order continuous differentiable with respect to $y$ and  first order continuous differentiable with respect to $v$.
 Moreover,
\be\label{martingale-cond-infinite}
\int_t^Te^{\int_t^s\left(\frac{\gamma}{1-\gamma} r(u)+\frac{1}{2}\frac{\gamma}{(1-\gamma)}\theta(u)^2\right)du}\sigma_F( \tilde Y(s),V(s) )\partial_y\tilde\eta(s,\tilde Y(s), V(s))dW(s)
\en
is a martingale to guarantee that \eqref{eta-infinite} is the solution for $\eta$. The optimal strategy is given by $\hat{\pi}(t, Y(t), V(t))$ defined in \eqref{optimal-multi-sol-infinite}.  
\end{proof}


\subsubsection*{Financial Implications}
Based on  \eqref{optimal-multi-sol-infinite}, we observe that the first term is proportional to the Sharpe ratio and the second term is driven by the ratio of $\sigma_F$ and $\sigma$. Hence, in the case of $\partial_y\eta >0$, the investment for risky assets increases in the factor volatility $\sigma_F$. For example, in the domestic stock market, if we treat the factors as foreign stocks, the larger volatility of foreign stocks implies the optimal proportion of the wealth in risky assets becomes larger. 

We now study one particular case of the infinite delay time, assuming $\sigma_F$ being a constant and 
\ba
\nonumber b(t)+\frac{\gamma}{1-\gamma}\theta(t) \sigma_F &=&\alpha_1Y(t)+\alpha_2V(t) ,\\
\nonumber\gamma r(t)+\frac{1}{2}\frac{\gamma}{1-\gamma}\theta(t)^2&=&\beta_1Y(t)^2+\beta_2V(t)^2 ,\\
h(Y(t))&=&Y(t). \label{sufficient-delay}
\ea 
with the parameter constraint $\beta_1<0$. We make an ansatz for $\eta$ written as 
\[
\eta(t,y,v)= \frac{1}{2}\psi_1(t)y^2+\frac{1}{2}\psi_2(t)v^2+\psi_3(t)yv+ \psi_4(t),
\]
where $\psi_i(t)$ for $i=1,\cdots,4$ are deterministic functions driven by $t$ satisfying
\ba
\label{psi1-delay}\dot\psi_1(t)&=&-2\alpha_1\psi_1(t)-\frac{\sigma_F^2}{1-\gamma}\psi_1(t)^2-2\psi_3(t)-2\beta_1,\\
\dot\psi_2(t)&=&2\lambda\psi_2(t)-2\alpha_2\psi_3(t)-\frac{\sigma_F^2}{1-\gamma}\psi_3(t)^2-2\beta_2,\\
\dot\psi_3(t)&=&(-\alpha_1+\lambda)\psi_3(t)-\frac{\sigma_F^2}{1-\gamma}\psi_1(t)\psi_3(t)-\alpha_2\psi_1(t)-\psi_2(t),\\
\dot\psi_4(t)&=&-\frac{1}{2}\sigma_F\psi_1(t),\label{psi4-delay}
\ea
with the terminal conditions $\psi_i(T)=0$ for $i=1,\cdots,4$ by identifying the coefficients of $y^2$, $v^2$, $yv$, and the zero order term. The existence and uniqueness of the coupled ordinary differential equations can be verified based on the condition $\beta_1<0$.  The corresponding first order derivative with respect to $y$ is given by
\[
\partial_y\eta(t,y,v)=\psi_1(t)y+\psi_3(t)v.
\]
the corresponding optimal strategy is given by $\hat{\pi}(t, Y(t), V(t))$ written as
\be 
\hat\pi(t,y,v)=\frac{\mu(y,v)-r(y,v)}{(1-\gamma)\sigma^2(y,v)}+\frac{\sigma_F (\psi_1(t)y+\psi_3(t)v)}{(1-\gamma)\sigma(y,v)  }.
\en

 In the case of $\beta_1<0$ and $\alpha_1>0,\; \alpha_2>0,\; \beta_2<0$, $\psi_1(t)$ and $\psi_3(t)$ stay in negative for $0\leq t\leq T$. See Figure \ref{psi-pointwise-delay} for instance. Therefore, in the case of $y<0$ and $v<0$, investors intend to invest risky assets rather than risk-free assets based on the larger proportion of the wealth to hold risky assets since $\partial_y\eta(t,y,v) >0$ for $0\leq t\leq T$ leading to
\[
\frac{\sigma_F \partial_y\eta(t,y,v)}{(1-\gamma)\sigma(y,v)}>0.
\]
However, given $y>0$ and $v>0$, the proportion of the wealth to hold risky assets $ \hat{\pi}(t, Y(t), V(t), Z(t))$ becomes smaller owing to $\partial_y\eta(t,y,v) <0$ implying
\[
\frac{\sigma_F \partial_y\eta(t,y,v)}{(1-\gamma)\sigma(y,v)}<0.
\]

 \begin{figure}[htbp]
\begin{center}
\includegraphics[width=12cm,height=5cm]{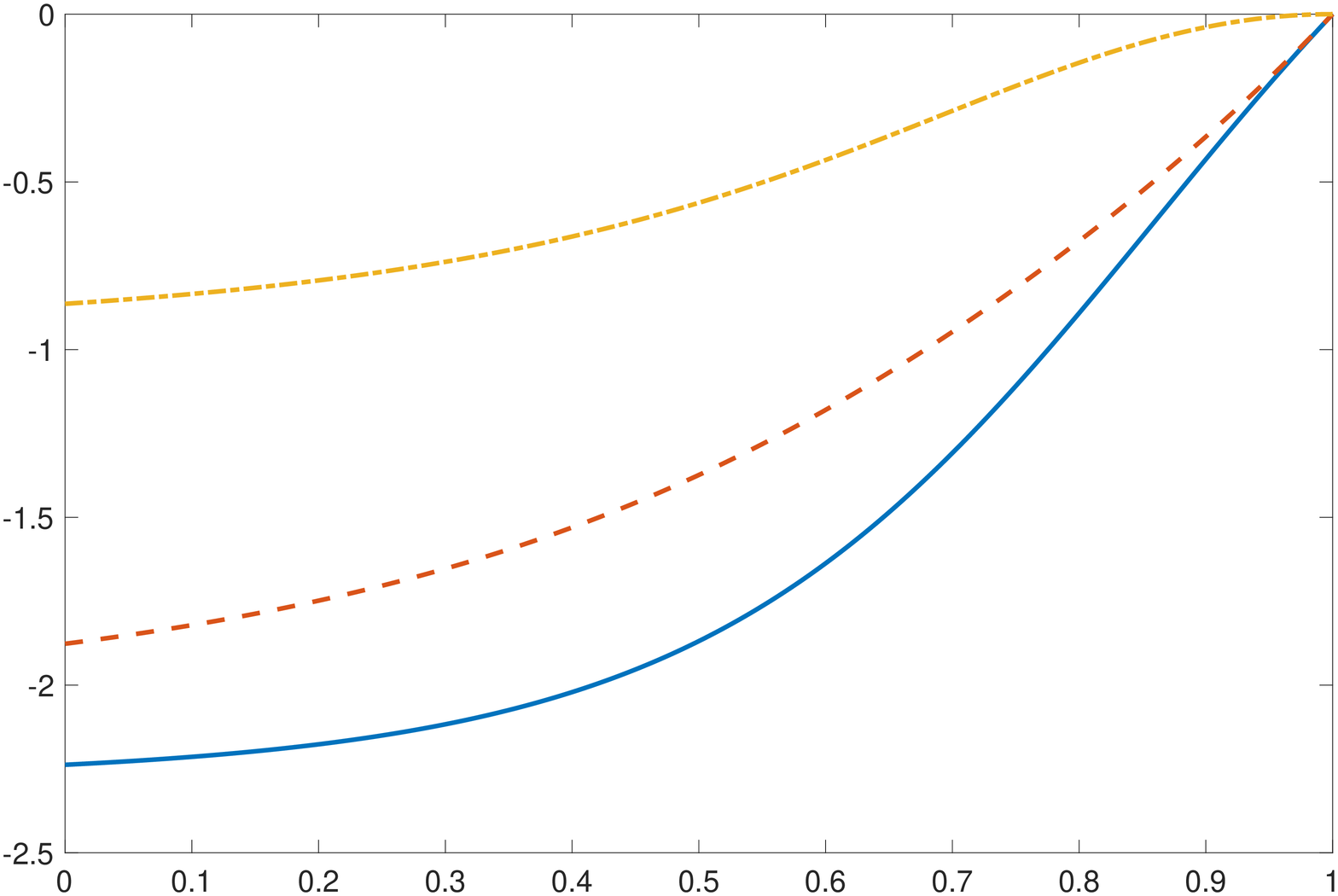}
\caption{Plot of $\psi_1$ (solid line), $\psi_2$ (dash line), and $\psi_3$ (dash-dot line) for the case of the infinite delay time with parameters $\alpha_1=1$, $\alpha_2=1$, $\beta_1=-2$, $\beta_2=-2$,  $\lambda=1$, $\sigma_F=1$, $\gamma=0.5$ and the terminal time $T=1$. }
\label{psi-pointwise-delay}
\end{center}
\end{figure}

\subsection{Pointwise Delay}
 Based on \cite{Pang2011} and \cite{Pang2015}, we assume $N=1$, $\sigma_F$ being a constant, and 
 \ba
\nonumber b(t)+\frac{\gamma}{1-\gamma}\theta(t)\sigma_F&=&\alpha_1Y(t)+\alpha_2V(t)+\alpha_3Z(t),\\
\nonumber \gamma r(t)+\frac{1}{2}\frac{\gamma}{(1-\gamma)}\theta(t)^2&=&\beta_1Y(t)+\beta_2V(t)+\beta_3Z(t),\\
h(Y(t))&=&Y(t),\label{assumption-special-1}
\ea
and $\alpha_i$ and $\beta_i$ for $i=1,\cdots,3$ satisfy 
\be\label{suff-cond-pw-delay}
\alpha_1+e^{\lambda\delta}\alpha_3=-\alpha_1\frac{\beta_3}{\alpha_3}+\beta_1=1, \quad \alpha_2-\lambda e^{\lambda\delta}\alpha_3=-\alpha_2\frac{\beta_3}{\alpha_3}+\beta_2=\alpha_3e^{\lambda \delta}.
\en
The optimal strategy is obtained  as follows.

\begin{prop}
Given the assumption \eqref{assumption-special-1}, the solution for the  $(\hat p,\hat q)$ is given by
\ba
\nonumber\hat p(t)&=&\left(e^{T-t}-1\right)\left(Y(t)+e^{\lambda\delta}\alpha_3V(t)\right)-\frac{\beta_3}{\alpha_3}Y(t)\\
\nonumber &&+\frac{\sigma_F^2}{2(1-\gamma)}\left\{\frac{1}{2}(e^{2(T-t)}-1)+2\left(1+\frac{\beta_3}{\alpha_3}\right)(1-e^{T-t})+\left(1+\frac{\beta_3}{\alpha_3}\right)^2(T-r)\right\},\\
 \label{hatp-delay}\\
\label{hatq-delay}\hat q(t)&=&\sigma_F\left(e^{T-t}-1-\frac{\beta_3}{\alpha_3}\right).
\ea
The optimal strategy for the optimization problem is $ \hat{\pi}(t, Y(t), V(t), Z(t))$ given by
\be\label{optimal-pointwise}
\hat\pi(t,y,v,z)=\frac{\mu(y,v,z)-r(y,v,z)}{(1-\gamma)\sigma^2(y,v,z)}+\frac{\sigma_F (e^{T-t}-1-\frac{\beta_3}{\alpha_3}) }{(1-\gamma)\sigma(y,v,z)  }. 
\en

\end{prop}

\begin{proof}
Based on the assumption \eqref{assumption-special-1}, assuming $\eta(t,Y(t), V(t))$ independent of $Z(t)$, the ansatz for $\hat p(t)$ is given by 
\[
\hat p(t)=\eta(t,Y(t),V(t)).
\]
Hence, applying It\^o formula to $\hat p(t)$, we get 
\ba
\nonumber d\hat p(t)&=&\bigg\{\partial_t\eta(t,Y(t),V(t))+b(t)\partial_y\eta(t,Y(t),V(t))+\frac{1}{2}\sigma_F  ^2\partial_{yy}\eta(t,Y(t),V(t))\bigg\}dt\\
\nonumber&&+(Y(t)-e^{-\lambda\delta}Z(t)-\lambda V(t))\partial_v\eta(t,Y(t),V(t))dt\\
&&+\sigma_F \partial_y\eta(t,Y(t),V(t))dW(t).\label{BSDE-special-1}
\ea
Identifying \eqref{BSDE-special-1} and \eqref{BSDE-multi-1} implies 
\be\label{optimal_q-special}
\hat q(t) = \sigma_F  \partial_y\eta(t,Y(t),V(t)),
\en
and $\eta$ satisfies 
\ban
\nonumber &&\partial_t\eta(t,y,v)+\frac{1}{2}\sigma_F ^2\partial_{yy}\eta(t,y,v)+(\alpha_1y+\alpha_2v+\alpha_3z)\partial_y\eta(t,y,v)\\
\nonumber&&+(y-e^{\lambda\delta}z-\lambda v)\partial_v\eta(t,y,v)+\frac{1}{2}\frac{1}{1-\gamma}\sigma_F  ^2(\partial_y\eta(t,y,v))^2+\beta_1y+\beta_2v+\beta_3z=0,
 \ean
 with the terminal condition $\eta(T,y,v)=0$.
Consequently, 
\ban
\nonumber &&\partial_t\eta(t,y,v)+  (\alpha_1\partial_y\eta+\partial_v\eta+\beta_1)y+(\alpha_2\partial_y\eta-\lambda\partial_v\eta+\beta_2)v\\
&&+(\alpha_3\partial_y\eta-e^{-\lambda\delta}\partial_v\eta+\beta_3)z+\frac{1}{2}\frac{1}{1-\gamma}\sigma_F  ^2(\partial_y\eta(t,y,v))^2+\frac{1}{2}\sigma_F ^2\partial_{yy}\eta(t,y,v)=0,
 \ean
with the terminal condition $ \eta(T,y,v)=0$. We obtain the solution for $\eta(t,y,v)$ given by
\be
\eta(t,y,v)=Q(t)\left(y+e^{\lambda\delta}\alpha_3v\right)-\frac{\beta_3}{\alpha_3}y+\psi(t)
\en
where $Q(t)$ and $\psi(t)$ are deterministic functions of $t$ satisfying 
\be
\dot Q(t)=- Q(t)-1,\quad \dot\psi(t)=-\frac{\sigma_F^2}{2(1-\gamma)} \left(Q(t)-\frac{\beta_3}{\alpha_3}\right)^2,
\en 
with the terminal condition $Q(T)=0$ and $\psi(T)=0$ leading to the solutions given by
\ban
Q(t)&=&e^{T-t}-1,\\
 \psi(t)&=&\frac{\sigma_F^2}{2(1-\gamma)}\left\{\frac{1}{2}(e^{2(T-t)}-1)+2\left(1+\frac{\beta_3}{\alpha_3}\right)(1-e^{T-t})+\left(1+\frac{\beta_3}{\alpha_3}\right)^2(T-r)\right\}.
\ean
We then obtain the solution for  (\ref{hatp-delay}-\ref{hatq-delay}). Here we use the condition (\ref{suff-cond-pw-delay}). This leads to the optimal portfolio given by $\hat{\pi}(t, Y(t), V(t), Z(t))$ in (\ref{optimal-pointwise}). The last result follows from (\ref{optimal-thm1}) and the expression (\ref{optimal-pointwise}).

\end{proof}
 
\subsubsection*{Financial Implications}
Given the optimal strategy of the pointwise delay written as \eqref{optimal-pointwise}, the proportion for the risky assets grows linearly in the factor volatility if  
$e^{T-t}-1-\frac{\beta_3}{\alpha_3}>0$ depending only on $\alpha_3$ and $\beta_3$ being the coefficient of the pointwise delay term. Namely, assuming $t=0$, when the terminal time $T$ is large enough given by $T>\log (1+\frac{\beta_3}{\alpha_3})$, the investors prefer risky assets in the case of the factor volatility becoming larger. Obviously,   $ \frac{\beta_3}{\alpha_3}<0$ implies $\partial_y\eta>0$ for all $T>0$. In this case, the proportion of investment in risky assets increases in the factor volatility $\sigma_F$.

\section{Conclusions}\label{conclusion}
In this paper, we consider the portfolio optimization problem with the delay factor model based on the power utility. Due to the non-Markovian feature, we apply the stochastic maximum principle and the FBSDE approach to characterize the optimal strategy and prove the existence and uniqueness for the corresponding FBSDEs. We further propose a different approach using martingale method to solve the optimization problem in the complete market, and investigate its connection with the FBSDE approach. In addition, two explicitly solvable cases are also discussed. The proposed problem in this paper can be extended toward many directions including the general utility and insurance analysis. \\

 From the financial point of view, 
it would be of great interest to study the portfolio optimization problem through maximizing the exponential utility or the general utility where the assets are driven by the delay factor model. In addition, the one player optimization problem can be extended to the multiple players optimization problem with game feature using the relative performance proposed by \cite{Touzi2015}. From the mathematical point of view, under (non)-Markovian structure, the discussion of the existence and uniqueness of the corresponding coupled quadratic FBSDEs is also needed. In the Markovian cases, the solvability of HJB equations is also worth to study. We will pursue to address these problems in the future work.

 \appendix
 \numberwithin{equation}{section}

\makeatletter 
\newcommand{\section@cntformat}{Appendix \thesection:\ }
\makeatother

 \section{Proof of Theorem \ref{thm:FBSDE}} \label{Proof-FBSDE}
The quadratic FBSDE is
\begin{align*}
\left\{ \begin{array}{ll} &d\tilde X^{\hat\pi}(t)=\frac{\gamma}{1-\gamma}\tilde X^{\hat\pi}(t)\left\{ (1-\gamma)r(t)+\frac{1}{2}\left|\theta(t)\right|^2+\frac{1}{2}\hat{q}(t)^*\sigma(t)^*\left[\sigma(t)\sigma(t)^*\right]^{-1} \sigma(t)\hat{q}(t)\right\}dt\\ 
&\qquad \qquad +\frac{\gamma}{1-\gamma}\tilde X^{\hat\pi}(t)\left\{\theta(t)^*+\hat{q}(t)^*\sigma(t)^*\left[\sigma(t)\sigma(t)^*\right]^{-1} \sigma(t)\right\}dW(t)\ ,\\[2mm]
&d\hat{p}(t)=\bigg\{-\gamma r(t)-\frac{1}{2}\frac{\gamma}{1-\gamma}\left|\theta(t)\right|^2-\frac{1}{2} \hat{q}(t)^*(I_{n\times n}+\frac{\gamma}{1-\gamma}\sigma(t)^*\left[\sigma(t)\sigma(t)^*\right]^{-1} \sigma(t))\hat{q}(t)\\
&\qquad \qquad \qquad -\frac{\gamma}{1-\gamma}\theta(t)^*\hat{q}(t) \bigg\}dt +\hat{q}(t)^*dW(t)\ ,\\[2mm]
 &\hat{X}(0)= U(x),\quad \hat{p}(T)=0\ .
 \end{array} \right. 
\end{align*}
Define  the following notations
\begin{align*}
&f(t,\hat{q}(t)):=\gamma r(t)+\frac{1}{2}\frac{\gamma}{1-\gamma}\left|\theta(t)\right|^2\\
&+\frac{1}{2}\hat{q}(t)^*\bigg(I_{n\times n}+\frac{\gamma}{1-\gamma}\sigma(t)^*\left[\sigma(t)\sigma(t)^*\right]^{-1} \sigma(t)\bigg)\hat{q}(t)+\frac{\gamma}{1-\gamma}\theta(t)^*\hat{q}(t)\ .
\end{align*}
Then we can write the dynamics of $\hat{p}(t)$ as follows
\begin{align}\label{QBSDE}
\left\{ \begin{array}{ll}
&d\hat{p}(t)=-f(t,\hat{q}(t))dt +\hat{q}(t)^*dW(t)\ ,\\[2mm]
& \hat{p}(T)=0\ .
\end{array} \right. 
\end{align}
Denote $$\tilde\gamma=\frac{\gamma}{1-\gamma}.$$ Consider the probability measure
\begin{align*}
d\tilde{\mathbb{P}}=\mathcal{E}_T(\tilde\gamma\theta^*\cdot W)d\mathbb{P}
\end{align*}
where 
\be\label{cal-E}
\mathcal{E}_T(\tilde\gamma\theta^*\cdot W)=\exp\left(\int_0^T\tilde\gamma\theta(t)^*dW(t)-\frac{1}{2}\int_0^T\|\tilde\gamma\theta(t)^*\|^2dt\right)
\en
is the Dol$\acute{\text{e}}$ans-Dade exponential martingale and 
\begin{align*}
	\tilde{W}(t)=W(t)-\int_0^t\tilde\gamma\theta(s)^*ds\ ,
\end{align*}
is a $\tilde{\mathbb{P}}$-Brownnian motion. Then \eqref{QBSDE} is equivalent to
\begin{align}\label{QBSDE_Girsanov}
\left\{ \begin{array}{ll}
&d\hat{p}(t)=-\tilde{f}(t,\hat{q}(t))dt +\hat{q}(t)^*d\tilde{W}(t)\ ,\\[2mm]
& \hat{p}(T)=0\ ,
\end{array} \right. 
\end{align}
where
\begin{align*}
\tilde{f}(t,\hat{q}(t))=&\gamma r(t)+\frac{1}{2} \tilde\gamma\left|\theta(t)\right|^2+\frac{1}{2}\hat{q}(t)^*\bigg(I_{n\times n}+\tilde\gamma\sigma(t)^*\left[\sigma(t)\sigma(t)^*\right]^{-1} \sigma(t)\bigg)\hat{q}(t)\\
=&\gamma r(t)+\frac{1}{2}\tilde\gamma\left|\theta(t)\right|^2+\frac{1}{2}\hat{q}(t)^*\bigg(I_{n\times n}+\tilde\gamma\tilde{\sigma}(t)\bigg)\hat{q}(t)
\end{align*}
 and
$$\tilde{\sigma}(t):=\sigma(t)^*\left[\sigma(t)\sigma(t)^*\right]^{-1} \sigma(t)\ .$$
Showing the system \eqref{QBSDE} has a unique solution is equivalent to showing \eqref{QBSDE_Girsanov} has a unique solution. Note that \eqref{QBSDE_Girsanov} is equivalent to the following system:
	\begin{align}\label{QBSDE_xi}
	\left\{ \begin{array}{ll}
	&d\tilde{p}(t)=-(\tilde{f}(t,\hat{q}(t))-\tilde{f}(t,0))dt +\hat{q}(t)^*d\tilde{W}(t)\\[3mm]
	&\qquad \ =-\frac{1}{2}\hat{q}(t)^*\bigg(I_{n\times n}+\tilde\gamma\sigma(t)^*\left[\sigma(t)\sigma(t)^*\right]^{-1} \sigma(t)\bigg)\hat{q}(t)dt+\hat{q}(t)^*d\tilde{W}(t)\ ,\\[2mm]
	& \tilde{p}(T)=\xi\ .
	\end{array} \right. 
	\end{align}	
	where $$
\tilde{p}(t)=\hat{p}(t)+\int_0^t \tilde{f}(s, 0) ds,
$$
and
$$\xi:=\int_0^T\tilde{f}(s,0)ds=\int_0^T(\gamma r(s)+\frac{1}{2}\tilde\gamma\theta(s)^*\theta(s))ds\ . $$
We prove that \eqref{QBSDE_xi} has a unique solution for bounded terminal condition $\|\xi\|_{\infty}<\infty$.

\subsection{Local Solution}
The first result states that when $\|\xi\|_{\infty}$ is small, the system \eqref{QBSDE_xi} has a unique solution.
\begin{prop}\label{prop:local}
Let $\beta^2:= 4\cdot \|I_{n\times n}-\tilde{\sigma}\|^2_{\infty}$. Suppose $\|\xi\|_{\infty}<\frac{1}{4\beta}$, there exists a unique solution $(\tilde{p},\hat{q})$ for the system \eqref{QBSDE_xi}  with the norm $\|\tilde{p}\|_{\infty}^2+ \|\hat{q}\cdot \tilde{W}\|^2_{BMO_T(\tilde{\mathbb{P}})}\leq \frac{1}{2\beta^2}$.	
\end{prop}
\begin{proof}
	Consider the mapping  $(\tilde{p},\hat{q})=F(\bar{p},\bar{q})$ by the relation
	\begin{align}\label{map:F}
	\left\{ \begin{array}{ll}
	&d\tilde{p}(t)=-(\tilde{f}(t,\bar{q}(t))-\tilde{f}(t,0))dt +\hat{q}(t)^*d\tilde{W}(t)\ ,\\[2mm]
	& \tilde{p}(T)=\xi\ .
	\end{array} \right. 
\end{align}	
Using the Ito's formula for $|\tilde{p}(t)|^2$ we can obtain that	
\begin{align*}
|\tilde{p}(t)|^2=\|\xi\|_{\infty}^2+ 2\int_{t}^T\tilde{p}(s)(\tilde{f}(s,\bar{q}(s))-\tilde{f}(s,0))ds-\int_t^T|\hat{q}(s)|^2ds-2\int_t^T\tilde{p}(s)\hat{q}(s)^*d\tilde{W}(s)\ .
\end{align*}	
Let $\tau$ be a stopping time taking valued in $[0, T]$. Taking conditional expectation $\tilde{\mathbb{E}}[\cdot |\mathcal{F}_{\tau}]$ on both sides of above equation, where $\tilde{\mathbb{E}}[\cdot|\mathcal{F}_{\tau}]$ denotes the conditional expectation $w.r.t.$ the probability measure $\tilde{\mathbb{P}}$ and filtration $\mathcal{F}_{\tau}$, and using the inequality $2ab\leq \frac{1}{4}a^2+4b^2$, we can obtain that
\begin{align*}
&|\tilde{p}(t)|^2+\tilde{\mathbb{E}}\left[\int_{\tau}^T|\hat{q}(s)|^2ds|\mathcal{F}_{\tau}\right]\\
\leq &\|\xi\|^2_{\infty}+ 2\|\tilde{p}\|_{\infty}\cdot  \tilde{\mathbb{E}}\left[\int_{\tau}^T|\tilde{f}(s,\bar{q}(s))-\tilde{f}(s,0)|ds|\mathcal{F}_{\tau}\right]\\
\leq &\|\xi\|^2_{\infty}+\frac{1}{4}\|\tilde{p}\|_{\infty}^2+4\left(\tilde{\mathbb{E}}\left[\int_{\tau}^T|\tilde{f}(s,\bar{q}(s))-\tilde{f}(s,0)|ds|\mathcal{F}_{\tau}\right]\right)^2\ .
\end{align*}	
In the rest of the proof of this Proposition, we use the notation 
$$
\|N\|_{BMO_T}=\|N\|_{BMO_T(\tilde{P})}
$$
for a $\tilde{P}$-martingale $N$. Note that
\begin{align*}
\frac{1}{2}\left(\|\tilde{p}\|^2_{\infty}+\|\hat{q}\cdot \tilde{W}\|_{BMO_T}^2\right)\leq & \max\left(\|\tilde{p}\|^2_{\infty},\|\hat{q}\cdot \tilde{W}\|_{BMO_T}^2\right)\\
\leq &\text{ess}\sup_{[[0,T]]}\left\{|\tilde{p}(\tau)|^2+\tilde{\mathbb{E}}\left[\int_\tau^T|\hat{q}(s)|^2ds|\mathcal{F}_\tau\right]\right\}\\
\leq &  \|\xi\|^2_{\infty} + \frac 14 \|\tilde{p}\|_{\infty} + 4\cdot\text{ess} \sup_{[[0,T]]} \left(\tilde{\mathbb{E}}[\int_{\tau}^T |\tilde{f}(s, \bar{q}(s))-\bar{f}(s, 0)| ds| \mathcal{F}_{\tau}]\right)^2.
\end{align*}	
 The essential supremum is taken over  $\tau\in [[0, T]]$, the collection of stopping times  taking values in $[0, T]$. We obtain that
\begin{align*}
&\frac{1}{4} \|\tilde{p}\|^2_{\infty}+ \frac{1}{2}\|\hat{q}\cdot \tilde{W}\|_{BMO_T}^2\\
\leq& \|\xi\|^2_{\infty}+ 4\cdot  \text{ess}\sup_{[[0,T]]}\left\{\tilde{\mathbb{E}}\left[\int_\tau^T|\tilde{f}(s,\bar{q}(s))-\tilde{f}(s,0)|ds|\mathcal{F}_\tau\right]\right\}^2\\
\leq & \|\xi\|^2_{\infty}+4\cdot  \text{ess}\sup_{[[0,T]]}\left\{\tilde{\mathbb{E}}\left[\int_\tau^T\left( \frac{1}{2} |I_{n\times n}+\tilde\gamma\tilde{\sigma}(s)|\cdot|\bar{q}(s)|^2\right)ds|\mathcal{F}_\tau\right]\right\}^2\\
\leq & \|\xi\|^2_{\infty} + \|I_{n\times n}+\tilde\gamma\tilde{\sigma}\|^2_{\infty}\cdot \|\bar{q}\cdot \tilde{W}\|_{BMO_T}^4 \ \quad (\text{by definition $\tilde{\sigma}$ is uniformly bounded})\ .
\end{align*}	
Recall the definition 
$$\beta^2= 4\cdot \|I_{n\times n}+\tilde\gamma\tilde{\sigma}\|^2_{\infty}\ ,$$
therefore,	we can obtain that
\begin{align*}
& \|\tilde{p}\|^2_{\infty}+  \|\hat{q}\cdot \tilde{W}\|_{BMO_T}^2\leq \|\tilde{p}\|_{\infty} + 2 \|\hat{q}\cdot W\|_{BMO}^2\\
 &\leq 4\|\xi\|^2_{\infty}+ \beta^2\cdot \|\bar{q}\cdot \tilde{W}\|_{BMO_T}^4 \leq 4\|\xi\|^2_{\infty}+  \beta^2(\|\bar{p}\|_{\infty}^2+\|\bar{q}\cdot \tilde{W}\|_{BMO_T}^2)^2\ .
\end{align*}
We can pick $R$ such that
\begin{align*}
4\|\xi\|^2_{\infty}+\beta^2R^4\leq R^2\ ,
\end{align*}
if and only if $\|\xi\|_{\infty}\leq \frac{1}{4\beta}$.
For instance, 
\be\label{suff-R}
R=\frac{1}{\sqrt{2}\beta}
\en 
satisfies this quadratic inequality. Therefore the ball
\begin{align*}
\mathcal{B}_R=\left\{(\tilde{p}, \hat{q}\cdot \tilde{W})\in S^{\infty}\times BMO_T,   \  \|\tilde{p}\|^2_{\infty}+  \|\hat{q}\cdot \tilde{W}\|_{BMO_T}^2 \leq R^2  \right\}\ ,
\end{align*}
is such that $F(\mathcal{B}_R)\subset\mathcal{B}_R$.
\ \\

Similarly for $(\bar{p}^j,\bar{q}^j\cdot \tilde{B})\in\mathcal{B}_R, j=1,2$, using the notation $\delta \bar{p}=\bar{p}^1-\bar{p}^2$, $\delta \bar{q}=\bar{q}^1-\bar{q}^2$. Using Ito's lemma, we can obtain that	
	\begin{align*}
	&\|\delta \tilde{p}\|_{\infty}^2+\|\delta \hat{q}\cdot\tilde{B}\|_{BMO_T}^2\\
	\leq& 4\cdot  \text{ess}\sup_{[[0,T]]}\left\{\tilde{\mathbb{E}}\left[\int_\tau^T|\tilde{f}(s,\bar{q}^1(s))-\tilde{f}(s,\bar{q}^2(s))|ds|\mathcal{F}_\tau\right]\right\}^2\\
	=&  \text{ess}\sup_{[[0,T]]}\left\{\tilde{\mathbb{E}}\left[\int_\tau^T|\bar{q}^1(s)^*(I_{n\times n}+\tilde\gamma\tilde{\sigma}(s))\bar{q}^1(s)-\bar{q}^2(s)^*(I_{n\times n}+\tilde\gamma\tilde{\sigma}(s))\bar{q}^2(s)|ds|\mathcal{F}_\tau\right]\right\}^2\\
	= &\text{ess}\sup_{[[0,T]]}\left\{\tilde{\mathbb{E}}\left[\int_\tau^T|\delta\bar{q}(s)^*(I_{n\times n}+\tilde\gamma\tilde{\sigma}(s))(\bar{q}^1(s)+\bar{q}^2(s))|ds|\mathcal{F}_\tau\right]\right\}^2\\
	\leq & \|I_{n\times n}+\tilde\gamma\tilde{\sigma} \|^2_{\infty}\cdot  \text{ess}\sup_{[[0,T]]}\tilde{\mathbb{E}}\left[\int_\tau^T|\delta\bar{q}(s)^*|^2ds|\mathcal{F}_\tau\right] \cdot  \text{ess}\sup_{[[0,T]]}\tilde{\mathbb{E}}\left[\int_\tau^T(|\bar{q}^1(s)|+|\bar{q}^2(s)|)^2ds|\mathcal{F}_\tau\right]\\
		\leq & 2\|I_{n\times n}+\tilde\gamma\tilde{\sigma} \|^2_{\infty}\cdot  \text{ess}\sup_{[[0,T]]}\tilde{\mathbb{E}}\left[\int_\tau^T|\delta\bar{q}(s)^*|^2ds|\mathcal{F}_\tau\right]\\
		&\qquad  \cdot \left( \text{ess}\sup_{[[0,T]]}\tilde{\mathbb{E}}\left[\int_\tau^T|\bar{q}^1(s)|^2ds|\mathcal{F}_\tau\right]+\text{ess}\sup_{[[0,T]]}\tilde{\mathbb{E}}\left[\int_\tau^T|\bar{q}^2(s)|^2ds|\mathcal{F}_\tau\right]\right)\\
	\leq & \frac{1}{2}\beta^2 \cdot \|\delta\bar{q}\cdot \tilde{W}\|^2_{BMO_T} \cdot 2R^2= \beta^2R^2 \|\delta\bar{q}\cdot \tilde{W}\|^2_{BMO_T}\\
	\leq &\frac{1}{2}\|\delta\bar{q}\cdot \tilde{W}\|^2_{BMO_T} \leq \frac{1}{2}\left(\|\delta \bar{p}\|_{\infty}^2+\|\delta\bar{q}\cdot \tilde{W}\|^2_{BMO_T}\right)\ ,
	\end{align*}
	thus $F(\cdot,\cdot )$ forms a contracting mapping on $\mathcal{B}_R$. Here we use the definition of $R$ given earlier (\ref{suff-R}).Thus, by a standard result,  $F(\cdot, \cdot)$ has a unique fixed point which is the unique solution of (\ref{map:F}) .  
\end{proof}

\subsection{Global Solution}
We next show that \eqref{QBSDE_xi} has a unique solution for any bounded $\|\xi\|_{\infty}<\infty$. For any $\|\xi\|_{\infty}<\infty$, it can be represented as sum $$\xi=\sum_{i=1}^K\xi^{(j)}$$ with $\|\xi^{(j)}\|_{\infty}<\frac{1}{4\beta}$ and $K\in\mathbb{N}$ is a finite integer. Consider the following iterated system, where $j=1,2...,K$:
	\begin{align}\label{equation:iterate}
	\left\{ \begin{array}{ll}
	&d\tilde{p}^{(j)}(t)=-\left(\tilde{f}(t,\hat{q}^{(1)}(t)+...+\hat{q}^{(j)}(t))-\tilde{f}(t,\hat{q}^{(1)}(t)+...+\hat{q}^{(j-1)}(t))\right)dt +\hat{q}^{(j)}(t)^*d\tilde{W}(t)\ ,\\[2mm]
	& \tilde{p}^{(j)}(T)=\xi^{(j)}\ ,\\ [2mm]
	& \tilde{p}^{(0)}=0\ , \ \hat{q}^{(0)}=0\ .
	\end{array} \right.   
	\end{align}	
From Proposistion \ref{prop:local} we have known the existence of $(\tilde{p}^{(1)},\hat{q}^{(1)})$ provided $\|\xi^{(1)}\|_{\infty}\leq \frac{1}{4\beta}$.   Next we show that if $\|\xi^{(2)} \|_{\infty}\leq \frac{1}{4\beta}$,  the solution $(\tilde{p}^{(2)},\hat{q}^{(2)})$ for the following QBSDE exists:
	\begin{align}\label{equation:iterate_2}
	\left\{ \begin{array}{ll}
	&d\tilde{p}^{(2)}(t)=-\left(\tilde{f}(t,\hat{q}^{(1)}(t)+\hat{q}^{(2)}(t))-\tilde{f}(t,\hat{q}^{(1)}(t))\right)dt +\hat{q}^{(2)}(t)^*d\tilde{W}(t),\;0\leq t\leq T,\\[2mm]
	& \tilde{p}^{(2)}(T)=\xi^{(2)}.
	\end{array} \right.   
	\end{align}

\begin{prop}
Suppose $\|\xi^{(2)}\|_{\infty}<\frac{1}{4\beta}$, there exists a unique solution $(\tilde{p}^{(2)},\hat{q}^{(2)})$ for the system \eqref{equation:iterate_2}.
\end{prop}
\begin{proof}
Note that  
\begin{align*}
&\tilde{f}(t,\hat{q}^{(1)}(t)+\hat{q}^{(2)}(t))-\tilde{f}(t,\hat{q}^{(1)}(t))\\
=&\frac{1}{2}(\hat{q}^{(2)}(t)+\hat{q}^{(1)}(t))^*\left(I_{n\times n}+\tilde\gamma\tilde{\sigma}(t)\right)(\hat{q}^{(2)}(t)+\hat{q}^{(1)}(t))-\frac{1}{2}\hat{q}^{(1)}(t)^*\left(I_{n\times n}+\tilde\gamma\tilde{\sigma}(t)\right)\hat{q}^{(1)}(t)\\
=&\frac{1}{2}\hat{q}^{(2)}(t)^*\left(I_{n\times n}+\tilde\gamma\tilde{\sigma}(t)\right)\hat{q}^{(2)}(t)+\hat{q}^{(1)}(t)^*\left(I_{n\times n}+\tilde\gamma\tilde{\sigma}(t)\right)\hat{q}^{(2)}(t)\ .
\end{align*}
Then \eqref{equation:iterate_2} becomes
 	\begin{align}\label{equation:iterate_3}
 	\left\{ \begin{array}{ll}
 	&d\tilde{p}^{(2)}(t)=-\frac{1}{2}\hat{q}^{(2)}(t)^*\left(I_{n\times n}+\tilde\gamma\tilde{\sigma}(t)\right)\hat{q}^{(2)}(t)dt +\hat{q}^{(2)}(t)^*\left\{d\tilde{W}(t)-\left(I_{n\times n}+\tilde\gamma\tilde{\sigma}(t)\right)\hat{q}^{(1)}(t)dt\right\}\ ,\\[2mm]
 	& \tilde{p}^{(2)}(T)=\xi^{(2)}
 	\end{array} \right.   
 	\end{align}
 Using the Girsanov's theorem and Property 3 of BMO in Section 11.1.2 of \cite{touzi2012optimal}, we can define the probability $\mathbb{P}^{(2)}$ by
 $$d\mathbb{P}^{(2)}:=\mathcal{E}_T(\left(I_{n\times n}+\tilde\gamma\tilde{\sigma}\right)\hat{q}^{(1)}\cdot \tilde{W})d\tilde{\mathbb{P}}\ ,$$
 where ${\cal E}_T(\cdot)$ is defined in (\ref{cal-E}),   and the Brownian motion
 $$d\tilde{W}^{(2)}(t):=d\tilde{W}(t)-\left(I_{n\times n}+\tilde\gamma\tilde{\sigma}(t)\right)\hat{q}^{(1)}(t)dt\ $$
so that \eqref{equation:iterate_3} becomes
 	\begin{align}\label{equation:iterate_4}
 	\left\{ \begin{array}{ll}
 	&d\tilde{p}^{(2)}(t)=-\frac{1}{2}\hat{q}^{(2)}(t)^*\left(I_{n\times n}+\tilde\gamma\tilde{\sigma}(t)\right)\hat{q}^{(2)}(t)dt +\hat{q}^{(2)}(t)^*d\tilde{W}^{(2)}(t)\ ,\\[2mm]
 	& \tilde{p}^{(2)}(T)=\xi^{(2)}
 	\end{array} \right.   
 	\end{align}
which has the same form as \eqref{QBSDE_xi}. Using the same argument as that in Proposition \ref{prop:local}, we can show that under the condition $\|\xi^{(2)}\|_{\infty}\leq \frac{1}{4\beta}$, the system  \eqref{equation:iterate_4} has a unique solution $(\tilde{p}^{(2)},\hat{q}^{(2)})$ in the space $S_{\infty}([0,T]\times \Omega, \mathbb{R}^d)\times BMO_T(\mathbb{P}^{(2)})$.

\end{proof}
 Similarly, we can show that the solutions $(\tilde{p}^{(j)},\hat{q}^{(j)})$, $j=1,2,3,..,K$ exist for the system \eqref{equation:iterate}. Finally, we can see that 
 $$
\tilde{p}(t)=\tilde{p}^{(1)}(t)+ \tilde{p}^{(2)}(t)+\cdots+ \tilde{p}^{(K)}(t),
$$
 and 
$$
\hat{q}(t)=\tilde{q}^{(1)}(t)+ \tilde{q}^{(2)}(t)+\cdots+ \tilde{q}^{(K)}(t)
$$
solve \eqref{QBSDE_xi}.

\subsection{Global Uniqueness}
We consider the QBSDE \eqref{QBSDE}. We first find that
\begin{align*}
|f(t,\hat{q}(t))|\leq & \left|\gamma r(t)+\frac{1}{2}\tilde\gamma\left|\theta(t)\right|^2\right|+\frac{1}{2}\|\hat{q}(t)\|^2\cdot \left\|I_{n\times n}+\tilde\gamma\tilde{\sigma}\right\|+\|\tilde\gamma\theta(t)\|\cdot \|\hat{q}(t)\|\\
\leq &\left|\gamma r(t)+\frac{1}{2}\tilde\gamma\left|\theta(t)\right|^2\right|+\frac{1}{2}\|\tilde\gamma\theta(t)\|^2+\frac{1}{2}\|\hat{q}(t)\|^2\cdot (\left\|I_{n\times n}+ \tilde{\gamma} \tilde{\sigma}\right\|+1)\\
\leq &K+ M\cdot \|\hat{q}(t)\|^2
\end{align*}
where
\begin{align*}
&K:=\left\|\gamma r(\cdot)+\frac{1}{2}\tilde\gamma\left|\theta(\cdot)\right|^2\right\|_{\infty}+\frac{1}{2}\|\tilde\gamma\theta(\cdot)\|^2_{\infty}\ ,\\
&M:=\frac{1}{2}(\left\|I_{n\times n}+\tilde\gamma\tilde{\sigma}\right\|_{\infty}+1)\ .
\end{align*}

\begin{lemma}
Let $\xi\in L^{\infty}(\Omega, \mathbb{R})$. Assume that $(\hat p,\hat q)$ is a solution such that $\hat p$ is bounded. That is, $\hat p\in S_{\infty}([0,T]\times\Omega;\mathbb{R}^d)$. Then $\hat q\cdot W$ is in $BMO_T(\mathbb{P})$.
\end{lemma}
\begin{proof}
Let $\hat{p}$ be a solution of \eqref{QBSDE}, and there be a constant  $C>0$ such that
\begin{align*}
|\hat{p}(t)|\leq C\quad \text{a.s. for all}\ t\ .
\end{align*} 	
Denote $\tau$ as a stopping time such that $0\leq\tau\leq T$. Let $\lambda$ be a fixed real number (to be chosen later). Applying the Ito's lemma for $\exp(\lambda \hat{p}(T))-\exp(\lambda\hat{p}(\tau))$ and using the boundary condition $\hat{p}(T)=0$, we have
\ban
&&\frac{\lambda^2}{2}\int_{\tau}^{T}e^{\lambda \hat{p}(s)}|\hat{q}(s)|^2ds-\lambda\int_{\tau}^{T}e^{\lambda \hat{p}(s)} f(s,\hat{q}(s))ds+\lambda\int_{\tau}^{T}e^{\lambda \hat{p}(s)}\hat{q}(s)^*dW_s\\
&&=\exp(\lambda \hat{p}(T))-\exp(\lambda \hat{p}(\tau))\leq 1\ .
\ean	 
 
If $\hat{q}^*\cdot W$ is square integrable martingale, taking conditional expectation in above inequality, we obtain
\begin{align*}
&\frac{\lambda^2}{2}\mathbb{E}\left[\int_{\tau}^{T}e^{\lambda \hat{p} (s)}|\hat{q}(s)|^2ds\bigg| \mathcal{F}_\tau\right]\\
\leq & 1+\lambda\cdot  \mathbb{E}\left[\int_{\tau}^{T}e^{\lambda \hat{p}(s)} f(s,\hat{q}(s))ds\bigg|\mathcal{F}_\tau\right]\\
\leq &1+\lambda K\cdot  \mathbb{E}\left[\int_{\tau}^{T}e^{\lambda \hat{p}(s)}ds\bigg|\mathcal{F}_\tau\right]+\lambda M\cdot  \mathbb{E}\left[\int_{\tau}^{T}e^{\lambda \hat{p}(s)}|\hat{q}(s)|^2ds\bigg|\mathcal{F}_\tau\right]
\end{align*}
which implies that
\begin{align*}
\left(\frac{\lambda^2}{2}-\lambda M\right)\mathbb{E}\left[\int_{\tau}^{T}e^{\lambda \hat{p}(s)}|\hat{q}(s)|^2ds\bigg|\mathcal{F}_\tau\right]\leq 1+\lambda K\cdot  \mathbb{E}\left[\int_{\tau}^{T}e^{\lambda \hat{p}(s)}ds\bigg|\mathcal{F}_\tau\right]\ .
\end{align*}
Taking $$\lambda= 4M\ ,$$
we can have
\begin{align*}
4M^2\cdot \mathbb{E}\left[\int_{\tau}^{T}e^{\lambda \hat{p}(s)}|\hat{q}(s)|^2ds\bigg| \mathcal{F}_\tau\right]\leq 1+\lambda KT\cdot  e^{4MC}\ ,  
\end{align*} 
for any stopping time $\tau$. Since $\hat{p}\geq -C$, then we can obtain that
\begin{align*}
\|\hat{q}\cdot W\|^2_{BMO_T(\mathbb{P})}\leq \frac{ (1+\lambda KT\cdot  e^{4MC})e^{4MC}}{4M^2} \ .
\end{align*}
Finally, we can conclude that the solution to \eqref{QBSDE} belongs to $S_{\infty}([0,T]\times \Omega;\mathbb{R}^d)$ and $\hat{q}\cdot W \in BMO_T(\mathbb{P})$.

 
\end{proof}	

\begin{theorem}\label{theorem-uniqueness}
	There exists unique solution of \eqref{QBSDE}  such that $\hat{p}\in S_{\infty}([0,T]\times \Omega; \mathbb{R}^d)$ and $\hat{q}\cdot W \in BMO_T(\mathbb{P})$.
\end{theorem}
\begin{proof}
Suppose that $(\hat{p},\hat{q})$ and $(\hat{p}_1,\hat{q}_1)$ are two solutions to \eqref{QBSDE}. Let
$$\delta p=\tilde{p}-\hat{p}_1$$
and 
$$\delta q=\hat{q}-\hat{q}_1\ .$$
 We have
\begin{align*}
\delta p(t)=\delta p(0)-\int_0^t\left[f(s,\hat{q}(s))-f(s,\hat{q}_1(s))\right]ds+\int_0^t\delta q(s)^*dW(s)\ .
\end{align*}
Let
\begin{align*}
&\partial_jf(s)\equiv\partial_jf(s,\hat{q}(s),\hat{q}_1(s))\\
:=&\frac{f(s,\hat{q}^1(s),...,\hat{q}^{j-1}(s),\hat{q}^j(s),\hat{q}_1^{j+1}(s),...,\hat{q}_1^{n}(s))-f(s,\hat{q}^1(s),...,\hat{q}^{j-1}(s),\hat{q}_1^{j}(s),\hat{q}_1^{j+1}(s),...,\hat{q}_1^{n}(s))}{\hat{q}^{j}(s)-\hat{q}_1^{j}(s)}\ ,
\end{align*}
for $j\in\{1,...n\}$ where $\hat{q}^j$  and $\hat{q}_1^j$ are the $j^{th}$ component of $\hat{p}$ and $\hat{p}_1$. Let 
$$\nabla f(s)=(\partial_1f(s),...,\partial_nf(s))^*\ .$$
 Thus 
\begin{align*}
f(s,\hat{q}(s))-f(s,\hat{q}_1(s))=\nabla f(s)^*\delta {q}(s)\ .
\end{align*}
Then we have
\begin{align*}
d\delta p(s)=- \nabla f(s)^*\delta {q}(s)ds+\delta q(s) dW(s) \ , \ \delta p(T)=0\ .
\end{align*}
From Lemma 4.1,  $$\nabla f\cdot W\in BMO_T(\mathbb{P})$$
then by Property 3 of BMO in Section 11.1.2 of \cite{touzi2012optimal} and Girsanov Theorem, we define a new probability 
$$dQ=\mathcal{E}_T(\nabla f\cdot W)d\mathbb{P}\ ,$$
where ${\cal E}(\cdot)$ is defined in (\ref{cal-E}).
Then we obtain that $\delta p(s)$ is a $Q$-martingale, which yields
$$\delta p(s)=\mathbb{E}_Q[\delta p(T)|\mathcal{F}_s]=0\ , \ \forall s\in [0,T]\ .$$

\end{proof}

 \section{Martingale Method}\label{Martingale-App}
In the case of complete markets, we consider the portfolio optimization problem written as 
\[
\sup_{\pi}\EE[U(X^{\pi}(T))]
\]
where $U(\cdot)$ is denoted as the utility function. Namely, $U(x)$ is increasing, concave, and $ U'(\infty)=0$. The above problem can be rewritten as 
\[
\sup_{\pi}\EE[U(X^{\pi}(T))]=\sup_{\xi\in{\mathcal{C }}(x)}\EE[U(\xi)]
\]
with 
\[
H_0(t)=\exp\left(-\int_0^t  \theta(s)^*dW(s)-\frac{1}{2}\int_0^t\left| \theta(s)\right|^2ds-\int_0^tr(s)ds\right)
\]
with the notation
\[
 \theta(t)=\sigma(t)^{-1}\left(\mu(t)-r(t){\bf 1}\right).
\]
 where ${\mathcal{C}}(x)=\left\{\xi\in{\mathcal{F}}(T):\EE\left[H_0(T)\xi\right]\leq x\right\}$ given by the budget constraint. Hence, we have 
 \be\label{martingale-ineq}
 \sup_{\xi\in{{\mathcal{C }}(x)}}\bigg\{\EE[U(\xi)]+z(x-\EE\left[H_0(T)\xi\right])\bigg\}\\
 = xz+\sup_{\xi\in{{\mathcal{C }}(x)}}\EE[U(\xi)-zH_0(T)\xi].
 \en
The equation \eqref{martingale-ineq} attains its maximum at 
\[
\hat\xi=I(zH_0(T))
\]
where $I(\cdot)$ is denoted as the inverse of $ U'(\cdot)$. We further assume the following condition 
\[
\Lambda(z)=\EE[H_0(T)I(zH_0(T))]<\infty 
\]
for all $z>0$. Using $\Lambda(z)$ being strictly decreasing in $z$ with $\Lambda(0)=\infty$ and $\Lambda(\infty)=0$, there exists a unique $z={\mathcal{Z}}(x)$ such that $\Lambda({\mathcal{Z}}(x))=x$. We now define $$M(t)=\EE[H_0(T)I({\mathcal{Z}}(x)H_0(T))|{\mathcal{F}}_t]$$ where ${\mathcal{F}}_t$ is a filtration with probability measure $\PP$ and $\psi(s)$ by using the martingale representation 
\be\label{martingale-1}
M(t)=x+\int_0^t\psi(s)^*dW(s).
\en 
In addition, applying It\^o formula to $ H_0(t)X^{\hat\pi}(t)$ gives 
\be\label{martingale-2}
H_0(t) X^{\hat{\pi}}(t)=x+ \int_0^t H_0(s) X^{\hat{\pi}}(s)(\hat{\pi}(s)^* \sigma(s) -\theta(s)^*) dW(s),
\en
and $\hat{\pi}(t)$ is chosen such that $M(t)=H_0(t)X^{\hat{\pi}}(t)$.
By identifying \eqref{martingale-1} and \eqref{martingale-2}, we obtain 
\[
\hat\pi (t)=(\sigma(t)^{*})^{-1}\left(  \theta(t)+  \psi(t)M(t)^{-1}\right).
\]


\end{document}